\documentclass[11pt]{article}

\usepackage{amsmath,amsfonts,amssymb,amsthm}
\usepackage[margin=1in]{geometry}
\usepackage[colorlinks]{hyperref}
\usepackage{amsthm}
\usepackage{enumerate}

\usepackage{algorithmic}
\usepackage{algorithm2e}
\usepackage{centernot}
\usepackage{mathtools}
\usepackage{stmaryrd}
\usepackage{array}
\usepackage{enumitem}
\usepackage{comment}
\usepackage{tikz}
\usepackage{pgf}
\usepackage[T1]{fontenc}
\usepackage[utf8]{inputenc}
\usepackage{geometry}
\usepackage{lmodern}
\usepackage[toc,page]{appendix}
\usetikzlibrary{decorations.pathreplacing}

\usepackage[classfont=sanserif,langfont=sanserif,funcfont=sanserif]{complexity}

\newlang{\BPSAT}{BP\text{-}SAT}
\newlang{\CNFSAT}{CNF\text{-}SAT}
\newlang{\OVabreviation}{OV}
\newlang{\CircuitSAT}{CircuitSAT}

\newclass{\DC}{DEPTH2}

\newcommand{\problem}[1]{\textsc{#1}}
\newcommand{\newproblem}[2]{\newcommand{#1}{\problem{#2}}}

\newproblem{\editdist}{Edit Distance}
\newproblem{\OV}{Orthogonal Vectors}
\newproblem{\matchingTrian}{Matching Triangles}
\newproblem{\TriangleCol}{Triangle Collection}
\newproblem{\CNFSATproblem}{CNF-SAT}
\newproblem{\FD}{Frechet Distance}
\newproblem{\AlignmentProblem}{Alignment}
\newproblem{\LCS}{Longest Common Subsequence}
\newproblem{\DTW}{Dynamic Time Warping}
\newproblem{\APSP}{APSP}
\newproblem{\threeSUM}{3SUM}
\newproblem{\BPSATproblem}{BP-SAT}
\newproblem{\collision}{2-to-1 Collision}

\newcommand{\ndots}{\!\!\dots}

\newclass{\DCQSETH}{\mathrm{DEPTH2\text{-}QSETH}}
\newclass{\NCQSETH}{\mathrm{NC\text{-}QSETH}}
\newclass{\gQSETH}{\gamma\mathrm{\text{-}QSETH}}
\newclass{\SETH}{\mathrm{SETH}}

\newcommand{\encodingName}{\textsc{Matrix Encoding}}

\DeclareMathOperator{\truthtable}{tt}
\DeclareMathOperator{\desc}{desc}
\DeclareMathOperator{\qTimeWB}{qTimeWB_{\epsilon}}
\DeclareMathOperator{\qTimeBB}{qTimeBB_{\epsilon}}
\DeclareMathOperator{\Query}{Q_{\epsilon}}
\DeclareMathOperator{\propertyP}{P}
\DeclareMathOperator{\circuitC}{C}
\DeclareMathOperator{\algoA}{\mathrm{A}}
\DeclareMathOperator{\compressionOblivious}{\mathcal{CO}}
\DeclareMathOperator{\alignCost}{alignment\text{-}cost}

\DeclareMathOperator{\propertyPPedit}{PP_{edit}}
\DeclareMathOperator{\propPedit}{P_{edit}}
\DeclareMathOperator{\matrixM}{M}
\DeclareMathOperator{\setOfPaths}{PATHS}
\DeclareMathOperator{\pathP}{P}
\DeclareMathOperator{\pathR}{R}
\DeclareMathOperator{\pathCost}{cost}
\DeclareMathOperator{\minPathCost}{\Delta}
\DeclareMathOperator{\editCost}{edit\text{-}cost}
\DeclareMathOperator{\coarseAlignC}{C}
\DeclareMathOperator{\coarseAlignD}{D}
\DeclareMathOperator{\coarseAlignN}{N}
\DeclareMathOperator{\groupG}{G}
\DeclareMathOperator{\bad}{bad}

\newtheorem{definition}{Definition}
\newtheorem*{theorem*}{Theorem}
\newtheorem{theorem}{Theorem}
\newtheorem{fact}{Fact}
\newtheorem{lemma}{Lemma}
\newtheorem{conjecture}{Conjecture}
\newtheorem*{conjecture*}{Conjecture}
\newtheorem{corollary}{Corollary}

\hypersetup{citecolor=cyan}

\begin{document}


\title{The Quantum Strong Exponential-Time Hypothesis}

\author{Harry Buhrman${^*} {^\ddagger}$ \\ \href{mailto:harry.buhrman@cwi.nl}{harry.buhrman@cwi.nl} 
 \and Subhasree Patro${^*} {^\ddagger}$ \\ \href{mailto:subhasree.patro@cwi.nl}{subhasree.patro@cwi.nl}
   \and Florian Speelman $^\ddagger$ \\ \href{mailto:f.speelman@cwi.nl}{f.speelman@cwi.nl}\\
   $^*$University of Amsterdam, 
$^\ddagger$QuSoft, CWI Amsterdam}

\date{}
\maketitle

\begin{abstract}
The strong exponential-time hypothesis (SETH) is a commonly used conjecture in the field of  complexity theory. It states that CNF formulas cannot be analyzed for satisfiability with a speedup over exhaustive search. This hypothesis and its variants gave rise to a fruitful field of research, fine-grained complexity, obtaining (mostly tight) lower bounds for many problems in $\P$ whose unconditional lower bounds are hard to find. In this work, we introduce a framework of Quantum Strong Exponential-Time Hypotheses, as quantum analogues to SETH.

Using the QSETH framework, we are able to translate quantum query lower bounds on black-box problems to conditional quantum time lower bounds for many problems in  $\BQP$. As an example, we illustrate the use of the QSETH by providing a conditional quantum time lower bound of $\Omega(n^{1.5})$ for the Edit Distance problem. We also show that the $n^2$ SETH-based lower bound for a recent scheme for Proofs of Useful Work, based on the Orthogonal Vectors problem holds for quantum computation assuming QSETH, maintaining a quadratic gap between verifier and prover.
\end{abstract}

\section{Introduction}
\label{sec:intro}

There is a rich diversity of problems that can be solvable in polynomial time, some that have surprisingly fast algorithms, such as the computation of Fourier transforms or solving linear programs, and some for which the worst-case run time has not improved much for many decades. 
Of the latter category $\editdist$ is a good example: this is a problem with high practical relevance, and an $O(n^2)$ algorithm using dynamic programming, which is simple enough to be taught in an undergraduate algorithms course, has been known for many decades.
Even after considerable effort, no algorithm has been found that can solve this problem in fewer than $O(n^2 / \log^2 n)$ steps~\cite{EditDist-MasekPaterson-80}, still a nearly quadratic run time.

Traditionally, the field of complexity theory has studied the time complexity of problems in a relatively coarse manner -- the class $\P$, the problems solvable in polynomial time, is one of the central objects of study in complexity theory.

Consider $\CNFSATproblem$, the problem of whether a formula, input in conjunctive normal form, has a satisfying assignment.
What can complexity theory tell us about how hard it is to solve this problem?
For $\CNFSATproblem$, the notion of $\NP$-completeness gives a convincing reason why it is hard to find a polynomial-time algorithm for this problem: if such an algorithm is found, all problems in the complexity class $\NP$ are also solvable in polynomial time, showing $\P=\NP$.

Not only is no polynomial-time algorithm known, but (if the clause-length is arbitrarily large) no significant speed-up over the brute-force method of trying all $2^n$ assignments is known.
Impagliazzo, Paturi, and, Zane \cite{ETH-ImpagliazzoPaturi-01, ETH-ImpagliazzoPaturiZane-01} studied two ways in which this can be conjectured to be optimal. The first of which is called the \emph{Exponential-Time Hypothesis} (ETH).

\begin{conjecture}[Exponential-Time Hypothesis]
	\label{conj:ETH}
	There exists a constant $\alpha > 0$ such that $\CNFSATproblem$ on $n$ variables and $m$ can not be solved in time $O(m2^{\alpha n})$ by a (classical) Turing machine.
\end{conjecture}

This conjecture can be directly used to give lower bounds for many natural $\NP$-complete problems, showing that if ETH holds then these problems also require exponential time to solve.
The second conjecture, most importantly for the current work, is the \emph{Strong Exponential-Time Hypothesis} (SETH).

\begin{conjecture}[Strong Exponential-Time Hypothesis]
	\label{conj:SETH}
	There does not exist $\delta > 0$ such that $\CNFSATproblem$ on $n$ variables and $m$ clauses can be solved in $O(m2^{n(1-\delta)})$ time by a (classical) Turing machine.
\end{conjecture}

The strong exponential-time hypothesis also directly implies many interesting exponential lower bounds within $\NP$, giving structure to problems within the complexity class.
A wide range of problems (even outside of just $\NP$-complete problems) can be shown to require strong exponential time assuming SETH: for instance, recent work shows that, conditioned on SETH, classical computers require exponential time for so-called \emph{strong simulation} of several models of quantum computation~\cite{HMS-18,MT-19}.

Surprisingly, the SETH conjecture is not only a very productive tool for studying the hardness of problems that likely require exponential time, but can also be used to study the difficulty of solving problems within $\P$, forming a foundation for the field of \emph{fine-grained complexity}.
The first of such a SETH-based lower bound was given in \cite{CNFSATtoOV-RWilliams-05}, via a reduction from $\CNFSATproblem$ to the $\OV$ problem, showing that a truly subquadratic algorithm that can find a pair of orthogonal vectors among two lists would render SETH false.

The $\OV$ problem became one of the central starting points for proving SETH-based lower bounds, and conditional lower bounds for problems such as computing the Frechet distance between two curves \cite{FrechetDist-Bringmann-14}, sequence comparison problems such as the string alignment problem \cite{ConsequencesOFAlignment-AbboudEtAl-14}, Longest Common Subsequence and Dynamic Time Warping \cite{TightHardness-AbboudEtAl-15}, can all obtained via a reduction from $\OV$.
Also the $\editdist$ problem~\cite{EditDist-BackursIndyk-15} can be shown to require quadratic time conditional on SETH, implying that any super-logarithmic improvements over the classic simple dynamic programming algorithm would also imply better algorithms for satisfiability -- a barrier which helps explain why it has been hard to find any new algorithms for this problem.

Besides $\CNFSATproblem$, the conjectured hardness of other key problems like $\threeSUM$ and $\APSP$ is also commonly used to prove conditional lower bounds for problems in $\P$.
See the recent  surveys \cite{Survey-VVWilliams-15,Survey-Vassilevska-Williams-18} for an overview of the many time lower bounds that can be obtained when assuming only the hardness of these key problems.

All these results give evidence for the hardness of problems relative to classical computation, but interestingly SETH does not hold relative to \emph{quantum} computation.
Using Grover's algorithm~\cite{Grover-96,QCT-BernsteinVazirani-97}, quantum computers are able to solve $\CNFSATproblem$ (and more general circuit satisfiability problems) in time $2^{n/2}$, a quadratic speedup relative to the limit that SETH conjectures for classical computation.

Even though this is in violation of the SETH bound, it is not in contradiction to the concept behind the strong exponential-time hypothesis: the input formula is still being treated as a black box, and the quantum speedup comes `merely' from the general quadratic improvement in unstructured search\footnote{For unstructured search this bound is tight~\cite{StrengthWeaknessQC-BennettEtAl-97,Searching-BoyerEtAl-98}.
Bennett, Bernstein, Brassard, and Vazirani additionally show that with probability 1 relative to a random oracle all of NP cannot be solved by a bounded-error quantum algorithm in time $o(2^{n/2})$.)}.

It could therefore be natural to formulate the quantum exponential time hypothesis as identical to its classical equivalent, but with an included quadratic speedup, as a `basic QSETH'.
For some problems, such as $\OV$, this conjecture would already give tight results, since these problems are themselves amenable to a speedup using Grover's algorithm.
See for instance the Master's thesis \cite{QAlgForSETHbasProbs-Jorg-19} for an overview of some of the SETH-based lower bounds that are violated in the quantum setting.

On the other hand, since the conditional lower bound for all problems are a quadratic factor lower than before, such a `basic QSETH' lower bound for $\editdist$ would be merely linear.
Still, the best currently-known quantum algorithm that computes edit distance takes quadratic time, so we would lose  some of the explanatory usefulness of SETH in this translation to the quantum case.

In this work, we present a way around this limit.
Realize that while finding a single marked element is quadratically faster for a quantum algorithm, there is no quantum speedup for many other similar problems.
For instance, computing whether the number of marked elements is odd or even can not be done faster when allowing quantum queries to the input, relative to allowing only classical queries~\cite{QPolyMethod-BealsEtAl-01,Parity-FarhiEtAl-98}. 

Taking the edit distance again as an illustrative example, after careful inspection of the reductions from $\CNFSATproblem$ to $\editdist$~\cite{QuadLowerBounds-BringmannKunnemann-15,EditDist-BackursIndyk-15,EditDist-AbboudEtAl-16}, we show that the result of such a reduction encodes more than merely the existence of an a satisfying assignment.
Instead, the result of these reductions also encodes whether \emph{many} satisfying assignments exist (in a certain pattern), a problem that could be harder for quantum computers than unstructured search.
The `basic QSETH' is not able to account for this distinction, and therefore does not directly help with explaining why a linear-time quantum algorithm for $\editdist$ has not been found.

We present a framework of conjectures, that together form an analogue of the strong exponential-time hypothesis: QSETH.
In this framework, we account for the complexity of computing various properties on the set of satisfying assignments, giving conjectured quantum time lower bounds for variants of the satisfiability problem that range from $2^{n/2}$ up to $2^n$.

\paragraph{Summary of results.}
\begin{itemize}
	\item We define the QSETH framework, connecting quantum query complexity to the proving of fine-grained (conditional) lower bounds of quantum algorithms.
	The framework encompasses both different properties of the set of satisfying assignments, and is also able to handle different input circuit classes -- giving a hierarchy of assumptions that encode satisfiability on CNF formulas, general formulas, branching programs, and so on.
	\item Some SETH-based $\Omega(T)$ lower bounds carry over to $\Omega(\sqrt{T})$ QSETH lower bounds, from which we immediately gain structural insight to the complexity class $\BQP$. 
	\item We show that, assuming QSETH, the \textit{Proofs of Useful Work} of Ball, Rosen, Sabin and Vasudevan~\cite{uPoW-BallRosenSabinVasudevan-17} require time $\widetilde{O}(n^2)$ to solve on a quantum computer, matching the classical complexity of these proofs of work.
	\item We prove that the $\editdist$ problem requires $\Omega(n^{1.5})$ time to solve on a quantum computer, conditioned on QSETH.
	We do this by showing that the edit distance can be used to compute a harder property of the set of satisfying assignments than merely deciding whether one satisfying assignment exists.
	
	Following~\cite{EditDist-AbboudEtAl-16}, we are able to show this for a version of QSETH where the input formulas are \emph{branching programs} instead, giving a stronger result than assuming the hardness for only CNF inputs.
	
	\item
	As a corollary to proof of the conditional edit-distance lower bound, we can show that the query complexity of the restricted Dyck language is linear for any $k=\omega(\log n)$, partially answering an open question posed by Aaronson, Grier, and Schaeffer~\cite{Trichotomy-AaronsonEtAl-19}.\footnote{Lower bounds for the restricted Dyck language  were recently independently proven by Ambainis, Balodis, Iraids, Pr\={u}sis, and Smotrovs~\cite{Dyck-AmbainisEtAl-19}, and  Fr\'ed\'eric Magniez~\cite{Dyck-Magniez-19}.}
\end{itemize} 

\paragraph{Related work.} Independently from this work, Aaronson, Chia, Lin, Wang, and Zhang~\cite{CP-AaronsonEtAl-19} recently also defined a basic quantum version of the strong exponential-time hypothesis, which assumes that a quadratic speed-up over the classical SETH is optimal.
They present conditional quantum lower bounds for $\OVabreviation$, the closest pair problem, and the bichromatic closest pair problem, by giving fine-grained quantum reductions to $\CNFSATproblem$.
All such lower bounds have a quadratic gap with the corresponding classical SETH lower bound.

Despite the overlap in topic, these results turn out to be complementary to the current paper:
In the current work we focus on defining a more extensive framework for QSETH that generalizes in various ways the basic version. Our more general framework can exhibit a quantum-classical gap that is less than quadratic, which allows us to give conditional lower bounds for edit distance ($\Omega(n^{1.5})$) and useful proofs of work (a quadratic gap between prover and verifier).
For our presented applications, the requirements of the fine-grained reductions are lower, e.g., when presenting a lower bound of $n^{1.5}$ for edit distance it is no problem if the reduction itself takes time $\widetilde{O}(n)$.\footnote{We use $\widetilde{O}$ to denote asymptotic behavior up to polylogarithmic factors.}
Conversely, we do not give the reductions that are given by~\cite{CP-AaronsonEtAl-19} -- those results are distinct new consequences of QSETH (both of the QSETH that is presented in that work, and of our more extensive QSETH framework).

\paragraph{Structure of the paper.} In Section~\ref{sec:QSETH} we motivate and state the QSETH framework. Following that, in Section~\ref{sec:ConsequuencesCNFQSETH} we present the direct consquences of QSETH, including the maintaining of some current bounds (with a quadratic loss), and the Useful Proof of Work lower bound.
In Section~\ref{sec:ConsequuencesNCQSETH} we present a conditional lower bound for the Edit Distance problem, and the lower bound to the restricted Dyck language we get as a corollary to the proof.
Finally, we conclude and present several open questions in Section~\ref{sec:Conclusion}.

\section{Defining the Quantum Strong Exponential-Time Hypothesis}
\label{sec:QSETH}
Almost all known lower bounds for quantum algorithms are defined in terms of \textit{query} complexity, which measures the number of times any quantum algorithm must access the input to solve an instance of a given problem.
There are two main methods in the field for proving lower bounds on quantum query complexity: The first one is the \textit{polynomial method}, based on the observation that the (approximate) degree of the unique polynomial representing a function is a lower bound on the number of queries any bounded-error quantum algorithm needs to make \cite{QPolyMethod-BealsEtAl-01}. The second main method is the \textit{adversary method} \cite{QLowerBounds-Ambainis-00}  which analyzes a hypothetical quantum adversary that runs the algorithm with a superposition of inputs instead of considering a classical adversary that runs the algorithm with one input and then modifies the input.

Despite the success of quantum query complexity and the fact that we know tight query lower bounds for many problems, the model does not take into account the computational efforts required after querying the input.
In particular, it is not possible to use query complexity to prove any lower bound greater than linear, since any problem is solvable in the query-complexity model after all bits are queried.
In general we expect the time needed to solve most problems to be much larger than the number of queries required for the computation, but it still seems rather difficult to formalize methods to provide unconditional quantum time lower bounds for explicit problems.
We overcome these difficulties by providing a framework of conjectures that can assist in obtaining \emph{conditional} quantum time lower bounds for many problems in $\BQP$.
We refer to this framework as the QSETH framework.

\paragraph{Variants of the classical SETH.}
The Strong Exponential-Time Hypothesis (SETH) was first studied \cite{ETH-ImpagliazzoPaturi-01, ETH-ImpagliazzoPaturiZane-01}, who showed that the lack of a $O(2^{n(1-\delta)})$ for a $\delta > 0$ algorithm to solve $\CNFSATproblem$ is deeply connected to other open problems in complexity theory.
Despite it being one the most extensively studied problems in the field of (classical) complexity theory, the best known classical algorithms for solving \textit{k}-SAT run in $2^{n-n/O(k)}m^{O(1)}$ time \cite{AlgoKsat-PaturiPudlakSaksZane-05}, while the best algorithm for the more-general $\CNFSATproblem$ is $2^{n-n/O(\log \Delta)}m^{O(1)}$ \cite{ClauseDensityCNFSAT-CalabroEtAl-06}, where $m$ denotes the number of clauses and $\Delta=m/n$ denotes the clause to variable ratio.

Even though no refutation of SETH has been found yet, it is plausible that the CNF structure of the input formulas does allow for a speed-up.
Therefore, if possible, it is preferable to base lower bounds on the hardness of more general kinds of (satisfiability) problems, where the input consists of wider classes of circuits.
For example, lower bounds based on $\NC$-SETH, satisfiability with $\NC$-circuits as input,\footnote{$\NC$ circuits are of polynomial size and polylogarithmic depth consisting of fan-in 2 gates.} have been proven for $\editdist$, $\LCS$ and other problems \cite{EditDist-AbboudEtAl-16}, in particular all the problems that fit the framework presented in~\cite{QuadLowerBounds-BringmannKunnemann-15}. 

Additionally, a different direction in which the exponential-time hypothesis can be weakened, and thereby made more plausible, is requiring the computation of different properties of a formula than whether at least one satisfying assignment exists.
For example, hardness of \emph{counting} the number of satisfying assignments is captured by \#ETH~\cite{DellEtAl-14}.
Computing existence is equivalent to computing the OR of the set of satisfying assignments, but it could also conceivably be harder to output, e.g., whether the number of satisfying assignments is odd or even or whether the number of satisfying assignments is larger than some threshold.
In the quantum case, generalizing the properties to be computed is not only a way to make the hypothesis more plausible: for many of such tasks it is likely that the quadratic quantum speedup, as given by Grover's algorithm, no longer exist.

\subsection{The basic QSETH}
\label{sec:QSETHframeworkBasic}
To build towards our framework, first consider what would be a natural generalization of the classical SETH. 

\begin{conjecture*}[Basic QSETH]
There is no bounded error quantum algorithm that solves $\CNFSATproblem$ on $n$ variables, $m$ clauses in $O(2^{\frac{n}{2}(1-\delta)}m^{O(1)})$ time, for any $\delta > 0$.
\end{conjecture*}

This conjecture is already a possible useful tool in proving conditional quantum lower bounds, as we present an example of this in Section~\ref{sec:LowerBoundDCQSETH}.\footnote{Additional examples of implications from such a version of QSETH can be found in the recent independent work of \cite{CP-AaronsonEtAl-19}.}

We first extend this conjecture with the option to consider wider classes of circuits.
Let~$\gamma$ denote a class of representations of computational models. Such a representation can for example be polynomial-size CNF formulas, polylog-depth circuits $\NC$, polynomial-size branching programs $\BP$, or the set of all polynomial-size circuits. The complexity of the latter problem is also often studied in the classical case, capturing the hardness of $\CircuitSAT$.

\begin{conjecture*}[Basic $\gQSETH$]
A quantum algorithm cannot, given an input $C$ from the set $\gamma$, decide in time $O(2^{\frac{n}{2}(1-\delta)})$  whether there exists an input $x \in \{0,1\}^n$ such that $C(x)=1$ for any $\delta>0$.
\end{conjecture*}

We also define $\DC$ for the set of all depth-2 circuits consisting of unbounded fan-in, consisting only of AND and OR gates.
This definition is later convenient when considering wider classes of properties, 
and it can be easily seen that `basic $\DCQSETH$' is precisely the `basic QSETH' as defined above.

Since both these basic QSETH variants already contain a quadratic speedup relative to the classical SETH, conditional quantum lower bounds obtained via these assumptions will usually also be quadratically worse than any corresponding classical lower bounds for the same problems.
For some problems, lower bounds obtained using the basic QSETH, or using $\gQSETH$ for a wider class of computation, will be tight.
However, for other problems no quadratic quantum speedup is known.

\subsection{Extending QSETH to general properties}
\label{sec:QSETHframework}

We  now extend the `basic $\gQSETH$' as defined in the previous section, to also include computing different properties of the set of satisfying assignments.
By extending QSETH in this way, we can potentially circumvent the quadratic gap between quantum and classical lower bounds for some problems.

Consider a problem in which one is given some circuit representation of a boolean function $f:\{0,1\}^n \rightarrow \{0,1\}$ and asked whether a property $\propertyP:\{0,1\}^{2^n}\rightarrow \{0,1\}$ on the truth table of this function evaluates to 1, that is, given a circuit $\circuitC$ the problem is to decide if $\propertyP(\truthtable(\circuitC))=1$, where $\truthtable(\circuitC)$ denotes the truth table of the boolean function computed by the circuit $\circuitC$.
If one can only access $\circuitC$ as a black box then it is clear that the amount of time taken to compute $\propertyP(\truthtable(\circuitC))$ is lower bounded by the number of queries made to the string $\truthtable(\circuitC)$. However, if provided with the description of $\circuitC$, which we denote by $\desc(\circuitC)$, then one can analyze $\circuitC$ to compute $\propertyP(\truthtable(\circuitC))$ possibly much faster.

For example, take the representation to be polynomial-sized CNF formulas and the property to be OR.
Then for polynomial-sized CNF formulas this is precisely the $\CNFSATproblem$ problem. Conjecturing quantum hardness of this property would make us retrieve the `basic QSETH' of the previous section.
Do note that we cannot simply conjecture that any property is hard to compute on CNF formulas:
Even though the query complexity of AND on a string of length $2^n$ is $\Omega(2^n)$ classically and $\Omega(2^{n/2})$ in the quantum case, this property can be easily computed in polynomial time both classically and quantumly when provided with the description of the $n^{O(1)}$ sized CNF formula.

To get around this problem, we can increase the complexity of the input representation:
If we consider inputs from $\DC$, the set of all depth-2 circuits consisting of unbounded fan-in AND and OR gates, we now have a class that is closed under complementation. For this class, it is a reasonable conjecture that both AND, the question whether the input is a tautology and all assignments are satisfying, and OR, the normal $\SAT$ problem, are hard to compute.

After this step we can look at further properties than AND and OR. For instance, consider the problem of computing whether there exists an even or an odd number of satisfying assignments. This task is equivalent to computing the PARITY of the truth table of the input formula.
How much time do we expect a quantum algorithm to need for such a task?

The quadratic speedup for computing the OR is already captured in the model where the quantum computation only tries possible assignments and then performs Grover's algorithm in a black box way.
If PARITY is also computed in such a way, then we know from query complexity~\cite{QPolyMethod-BealsEtAl-01} that there is no speedup, and the algorithm will have to use $\Omega(2^n)$ steps.
Our QSETH framework will be able to consider more-complicated properties, like PARITY.

Finally, observe that such a correspondence, i.e., between the query complexity of a property and the time complexity of computing this property on the set of satisfying assignments, cannot hold for \emph{all} properties, even when we consider more complicated input classes besides CNF formulas.
For instance, consider a property which is 0 on exactly the strings that are truth tables of polynomial-sized circuits, and is PARITY of its input on the other strings.
Such a property has high quantum query complexity, but is trivial to compute when given a polynomial-sized circuit as input.
We introduce the notion of \emph{compression oblivious} below to handle this problem.

\paragraph{Defining QSETH.} We formalize the above intuitions in the following way.
Let the variable $\gamma$ denote a class of representation at least as complex as the set $\DC$, where $\DC$ denotes the set of poly sized depth-2 circuits consisting of only OR and AND gates of unbounded fan-in. We define a meta-language $L_{\propertyP}$ such that $L_{\propertyP}=\{\desc(\circuitC)$ $|$ $\circuitC$ is an element from the set $\gamma$ and $\propertyP(\truthtable(\circuitC))=1\}$.
We now define the following terms:

\begin{definition}[White-box algorithms] An algorithm $\algoA$ decides the property $\propertyP$ in \textbf{white-box} if $\algoA$ decides the corresponding meta-language $L_{\propertyP}$. That is, given an input string $\desc(\circuitC)$, $\algoA$ accepts if and only if $\propertyP(\truthtable(\circuitC))=1$. We use $\qTimeWB(\propertyP)$ to denote the time taken by a quantum computer to decide the language $L_{\propertyP}$ with error probability $\epsilon$.
\end{definition}

\begin{definition}[Black-box algorithms]
An algorithm $\algoA$ decides the property $\propertyP$ in \textbf{black-box} if the algorithm $\algoA^f(1^n,1^m)$ accepts if and only if $\propertyP(\truthtable(f))=1$. Here, $f$ is the boolean function computed by the circuit $\circuitC$ and $m$ is the upper bound on $|\desc(\circuitC)|$ which is the size of the representation\footnote{For instance a CNF/DNF formula, an $\NC$ circuit, or a general circuit.} that describes $f$, and $\algoA^f$ denotes that the algorithm $\algoA$ has oracle access to the boolean function $f$.
We use $\qTimeBB(\propertyP)$ to denote the time taken by a quantum computer to compute the property $\propertyP$ in the black-box setting with error probability $\epsilon$.
\end{definition}

We define the set of \emph{compression oblivious} properties corresponding to $\gamma$ as the set of properties where the time taken to compute this property in the black-box setting is lower bounded by the quantum query complexity of this property on all strings. Formally,
\begin{equation*}
    \compressionOblivious(\gamma)=\{ \text{properties } \propertyP \text{ such that } \qTimeBB(\propertyP|_{\text{S}_{\gamma}}) \geq \Omega(\Query(\propertyP)) \},
\end{equation*}
where $\Query(\propertyP)$ denotes the quantum query complexity of the property $\propertyP$ in a $\epsilon$-bounded error query model and 
   $\text{S}_{\gamma}=\{ \truthtable(\circuitC) \text{ | } \circuitC \text{ is an element of the set } \gamma \}$.
For example, the properties~AND and OR are in $\compressionOblivious(\DC)$ because the adversarial set that gives the tight query bound for the property~AND (OR) are truth tables of functions that can be represented by $n^{O(1)}$ sized DNF (CNF) formulas.
As $\Query(\text{AND}|_{\text{S}_{\DC}})=\Query(\text{AND})$ and $\qTimeBB(\text{AND}|_{\text{S}_{\DC}}) \geq \Query(\text{AND}|_{\text{S}_{\DC}})$.
Therefore, we have $\text{AND} \in \compressionOblivious(\DC)$.
The same result holds for the property OR as well.

For each class of representation $\gamma$ we now define the corresponding $\gQSETH$, which states that computing any compression-oblivious property $\propertyP$ in the \textit{white-box} setting is at least as hard as computing $\propertyP$ in the \textit{black-box} setting. More formally,
\begin{conjecture}[$\gQSETH$]
For every class of representation $\gamma$, such as the class of depth-2 circuits $\DC$ or poly-sized circuits of a more complex class, for all properties $\propertyP \in \compressionOblivious(\gamma)$, we have $\qTimeWB(\propertyP|_{\gamma}) \geq \Omega(\Query(\propertyP))$. 
\end{conjecture}

\subsection{Observations on the set of compression oblivious properties}
As the class $\gamma$ gets more complex, the corresponding $\gQSETH$ becomes more credible.
The set of compression oblivious properties is an interesting object of study by itself.
First consider the following facts about sets of compression-oblivious properties relate, relative to different computational models.

\begin{fact}\label{fact:CO}
Given two classes of representations $A$ and $B$, if $A \subseteq B$ then for every property $\propertyP$, we have $\propertyP \in \compressionOblivious(\text{B})$ whenever $\propertyP \in \compressionOblivious(\text{A})$.
\end{fact}
\begin{fact}\label{fact:SETHinclusion}
Given two classes of representations $A$ and $B$, if $A \subseteq B$ then $A\mathrm{\text{-}QSETH}$ implies $B\mathrm{\text{-}QSETH}$.
\end{fact}
\begin{proof}
For Fact~\ref{fact:CO}. If $A \subseteq B$ then also for the corresponding sets of truth tables it holds that $\text{S}_\text{A} \subseteq \text{S}_\text{B}$.
If a property $\propertyP \in \compressionOblivious(\text{A})$, then $\qTimeBB(\propertyP|_{S_A}) \geq \Omega(\Query)(\propertyP)$ also implies $\qTimeBB(\propertyP|_{S_B}) \geq \qTimeBB(\propertyP|_{S_A})$ as $\text{S}_\text{B}$ is a superset of $\text{S}_\text{A}$. Therefore, $\propertyP \in \compressionOblivious(B)$.

For Fact~\ref{fact:SETHinclusion}: 
Whenever some property $\propertyP \in \compressionOblivious(\text{A})$ is hard to compute for inputs coming from $A$, this property is also $\propertyP \in \compressionOblivious(\text{B})$ by Fact~\ref{fact:CO}. Therefore, it is also hard to compute on an even wider range of inputs.
\end{proof}

Given an explicit property $\propertyP$ and a class of representation $\gamma$, it would be desirable to unconditionally prove that the property $\propertyP$ is \textit{$\gamma$-compression oblivious}\footnote{We call a property $\propertyP$ a $\gamma$-compression oblivious property if $\propertyP \in \compressionOblivious(\gamma)$.}.
This is possible for some simple properties that have query complexity $\Theta(\sqrt{N})$ like $\mathrm{OR}$, corresponding to ordinary satisfiability, and $\mathrm{AND}$.
Unfortunately, for more complicated properties, like computing the parity of the number of satisfying assignments, it turns out to be hard to find an unconditional proof that such a property is compression oblivious.
The following theorem shows a barrier to finding such an unconditional proof: proving that such a property is compression oblivious implies separating $\P$ from $\PSPACE$.


\begin{theorem}
\label{thm:PneqPSPACE}
If there exists a property $\propertyP$ such that $\Query(\propertyP)=\widetilde{\omega}(\sqrt{N})$ and $\propertyP$ is $\gamma$-compression oblivious, and for all $\epsilon>0$ we have $\propertyP \in \SPACE(N^\epsilon)$,  then $\P \neq \PSPACE$.
\end{theorem}
\begin{proof}[Proof (sketch)]
By way of contradiction, assume $\P = \PSPACE$.
We are given a promise that the circuit to whom we have black-box access to is in the set $\gamma$.
Using a simplified version of the algorithm for the oracle identification problem \cite{OracleIdentification-AmbainisEtAl-04,OIdent-kothari-14} and assuming $\P = \PSPACE$, we can extract a compressed form of the entire input using only $\widetilde{O}(\sqrt{N})$ quantum time.

As the property $\propertyP \in \SPACE(N^{\epsilon})$, using the $\P=\PSPACE$ assumption again, we can directly compute $\propertyP$ in time $O(N^{\epsilon})$ for any arbitrarily small $\epsilon$. Therefore, the total number of (quantum) steps taken is $\widetilde{O}(\sqrt{N}) + O(N^{\epsilon})$, which for an $\epsilon < \frac{1}{2}$ is in contradiction to the assumption that $\propertyP$ is $\gamma$-compression oblivious.
\end{proof}

An expanded version of the proof will be presented in a future version of the paper.
Note that SETH is already a much stronger assumption than $\P \neq \PSPACE$, therefore this observation leaves open the interesting possibility of proving that properties are compression oblivious assuming that the (Q)SETH holds for simpler properties. (For instance, these simpler properties could include OR and AND, for which it is possible to unconditionally prove that they are compression oblivious.)

\section{QSETH lower bounds for Orthogonal Vectors and Useful Proofs of Work}
\label{sec:ConsequuencesCNFQSETH}
Recall that $\DC$ denotes the set of polynomial-sized depth-2 circuits consisting of only OR and AND gates of unbounded fan-in.
Because of the simple input structure, the $\DCQSETH$ conjecture is therefore closest to the classical SETH, and implies the `basic QSETH' as introduced in Section~\ref{sec:QSETHframeworkBasic}:
\begin{corollary}
\label{cor:CNF-ORQSETH}
If $\DCQSETH$ is true then there is no bounded error quantum algorithm that solves $\CNFSATproblem$ on $n$ variables, $m$ clauses in $O(2^{(1-\delta)n/2}m^{O(1)})$ time, for any $\delta>0$.
\end{corollary}
\begin{proof}
Consider the property OR: $\{0,1\}^{2^n} \rightarrow \{0,1\}$. Using the fact that OR $\in \compressionOblivious(\DC)$, as shown in the previous section, we get $\qTimeWB(\text{OR}|_{\DC}) \geq \Omega(\Query(\text{OR}))=\Omega(2^{n/2})$.
Due to the structure of the DNF formulas one can compute the property OR on DNF formulas on $n$ variables, $m$ clauses in $n^{O(1)}m^{O(1)}$ time. This implies that the hard cases in the set $\DC$ for the OR property are the CNF formulas. Therefore, $\qTimeWB(\text{OR}|_{\text{CNF}}) \geq \Omega(2^{n/2})$ where the set CNF denotes all the polynomial sized CNF formulas.
\end{proof}

In this section we present several immediate consequences of the $\DCQSETH$ conjecture, including: 
\begin{enumerate}
    \item For some problems, classical $\SETH$-based $\Omega(T)$ time lower bounds carry over to the quantum case, with $\DCQSETH$-based $\Omega(\sqrt{T})$ quantum time lower bounds using (almost) the same reduction.
    \item The \textit{Proofs of Useful Work} of Ball, Rosen, Sabin and Vasudevan \cite{uPoW-BallRosenSabinVasudevan-17} require time $\widetilde{O}(n^2)$ to solve on a quantum computer, equal to their classical complexity, under $\DCQSETH$. 
\end{enumerate}

\subsection{Quantum time lower bounds based on \texorpdfstring{$\DCQSETH$}{DEPTH2-QSETH}}
\label{sec:LowerBoundDCQSETH}
The statement of $\DCQSETH$ along with Corollary \ref{cor:CNF-ORQSETH} can give quantum time lower bounds for some problems for which we know classical lower bounds under $\SETH$ (Conjecture \ref{conj:SETH}).
\begin{corollary}
\label{cor:SETH-CNF-OR-QSETH}
Let $\propertyP$ be a problem with an $\Omega(T)$ time lower bound modulo $\SETH$. Then, $\propertyP$ has an $\widetilde{\Omega}(\sqrt{T})$ quantum time lower bound conditioned under $\DCQSETH$ if there exists a classical reduction from $\CNFSAT$ to the problem $\propertyP$ taking $O(2^{\frac{n}{2}(1-\alpha)})$ (for $\alpha>0$) time or if there exists an efficient reduction that can access a single bit of the reduction output.\footnote{Note that $\SETH$ talks about solving $\CNFSATproblem$ as opposed to bounded $k$-SAT problems.
One could also define a quantum hardness conjecture for $k$-CNF or $k$-DNF, for an arbitrary constant $k$, in the same way as the original $\SETH$.
This variant is required for reductions that use the fact that $k$ is constant, which can occur through usage of the sparsification lemma~\cite{ETH-ImpagliazzoPaturi-01}. For examples where this is necessary within fine-grained complexity, see the \emph{Matching Triangles} problem mentioned in \cite{MatchingTrian-AbboudEtAl-15} or reductions like in \cite{probsHardCNFSAT-MarekEtAl-16}.}
\end{corollary}
As examples we will consider the $\OV$ and the $\editdist$ problem.
The $\OV$ ($\OVabreviation$) problem is defined as follows. Given two sets $U$ and $V$ of $n$ vectors, each over $\{0,1\}^d$ where $d=\omega(\log n)$, determine whether there exists a $u\in U$ and a $v \in V$ such that $\Sigma_{l \in [d]}u_lv_l=0$.
In \cite{CNFSATtoOV-RWilliams-05}, Williams showed that $\SETH$ implies the non-existence of a sub-quadratic classical algorithm for the $\OVabreviation$ problem.
In the quantum case the best-known query lower bound is $\Omega(n^{2/3})$, which can be achieved by reducing the $\collision$ problem to the $\OV$ problem; however, the known quantum time upper bound is $\widetilde{O}(n)$ \cite{QAlgForSETHbasProbs-Jorg-19}.
First note that we cannot use Williams' classical reduction directly, since a hypothetical quantum algorithm for $\OVabreviation$ expects quantum access to the input, and writing down the entire reduction already takes time $2^{n/2}$.
Instead, observe that the reduction produces a separate vector for each partial assignment: let $t(n)$ be the time needed to compute a single element of the output of the reduction, then $t(n) = \poly(n)$, which is logarithmic in the size of the total reduction.
Let $N=O^*(2^{n/2})$ be the size of the output of the reduction of~\cite{CNFSATtoOV-RWilliams-05}, for some CNF formula with $n$ variables.
Any quantum algorithm that solves $\OVabreviation$ in time $N^\alpha$, can solve $\CNFSAT$ in time $t(n) O^*(2^{\alpha n/2}) = O^*(2^{\alpha n/2})$.\footnote{We use $O^*$ to denote asymptotic complexity ignoring polynomial factors.}
Assuming $\DCQSETH$, this implies that a quantum algorithm requires time $\tilde{\Theta}(N)$ to solve $\OVabreviation$ for instances of size $N$.

See the recent results by Aaronson, Chia, Lin, Wang, and Zhang~\cite{CP-AaronsonEtAl-19} for more examples of reductions from (a variant of) QSETH, that also hold for the basic QSETH of our framework.
Additionally, there the authors define the notion of \emph{Quantum Fine-grained Reductions} more generally, and present a study of $\OVabreviation$ that also includes the case of constant dimension.

The next example we consider is the $\editdist$ problem. The $\editdist$ problem is defined as follows. Given two strings $a$ and $b$ over an alphabet set $\Sigma$, the edit distance between $a$ and $b$ is the minimum number of operations (insertions, deletions, substitutions) on the symbols required to transform string $a$ to $b$ (or vice versa).
A reduction by \cite{EditDist-BackursIndyk-15} shows that if the edit distance between two strings of length $n$ can be computed in time $O(n^{2-\delta})$ for some constant $\delta > 0$, then satisfiability on CNF formulas with $n$ variables and $m$ clauses can be computed in $O(m^{O(1)}\cdot 2^{(1-\frac{\delta}{2})n})$ which would imply that $\SETH$ (Conjecture~\ref{conj:SETH}) is false.
Just like in the $\OV$ case, we observe that the classical reduction from $\CNFSAT$ to $\editdist$ is local, in the sense that accessing a single bit of the exponentially-long reduction output can be done in polynomial time:
Every segment of the strings that are an output of the reduction, depend only on a single partial satisfying assignment, out of the $2^{n/2}$ possible partial assignments.

This observation directly lets us use the reduction of \cite{EditDist-BackursIndyk-15} to give a quantum time lower bound of $\widetilde{\Omega}(n)$ for the $\editdist$ problem, where $n$ here is the length of the inputs to $\editdist$, conditioned on $\DCQSETH$.
However, an unconditional quantum query lower bound of $\Omega(n)$ can also be easily achieved by embedding of a problem with high query complexity, such as the majority problem, in an edit distance instance.

We witness that with $\DCQSETH$ conjecture, the $\SETH$-based fine-grained lower bounds at best transfer to a square root lower complexity in the quantum case. This is definitely interesting on its own, but we are aiming for larger quantum lower bounds, which is why we focus on our more general framework.

\subsection{Quantum Proofs of Useful Work}
\label{sec:QuPoW}

Other applications of $\DCQSETH$ include providing problems for which \textit{Proofs of Useful Work (uPoW)} can be presented in the quantum setting. The paper \cite{uPoW-BallRosenSabinVasudevan-17} proposes uPoW protocols that are based on delegating the evaluation of low-degree polynomials to the prover. They present a classical uPoW protocol for the $\OV$ problem (OV) whose security proof is based on the assumption that OV needs $\Omega(n^{2-o(1)})$ classical time in the worst case setting, implying that the evaluation of a polynomial that encodes the instance of $\OVabreviation$ has average-case hardness.
At the end of this protocol, the verifier is able to compute the number of orthogonal vectors in a given instance.

Therefore, the same protocol also works to verify the solutions to $\oplus$OV, where $\oplus$OV denotes the parity version of OV, i.e., given two sets $U$, $V$ of $n$ vectors from $\{0,1\}^d$ each, output the parity of number of pairs $(u,v)$ such that $u \in U$, $v \in V$ and $\Sigma_{l\in [d]} u_lv_l=0$, where $d$ is taken to be $\omega(\log n)$.
Assuming $\DCQSETH$ and assuming PARITY $\in \compressionOblivious(\DC)$ we get that $\oplus \CNFSATproblem$ takes $\Omega(2^n)$ quantum time.
Due to the reduction\footnote{Note that here one can use the classical reduction from $\CNFSAT$ to $\OV$ that runs in $\widetilde{O}(2^{n/2})$.} given in \cite{CNFSATtoOV-RWilliams-05}, this protocol then implies a conditional quantum time lower bound of $\Omega(n^2)$ for the $\oplus$OV problem.
Therefore, the uPoW protocol by \cite{uPoW-BallRosenSabinVasudevan-17} also requires quantum provers to take time $\widetilde{\Omega}(n^2)$.

\section{Lower bound for edit distance using \texorpdfstring{$\NCQSETH$}{NC-QSETH}}
\label{sec:ConsequuencesNCQSETH}
In this section we discuss a consequence of our $\NCQSETH$ conjecture: Quantum time lower bound for the $\editdist$ problem.
Edit distance (also known as the \textit{Levenshtein distance}) is a measure of dissimilarity between two strings.
For input strings of length $n$, the well known Wagner--Fischer algorithm (based on dynamic programming) classically computes the edit distance in $O(n^2)$ time. Unfortunately, all the best known classical (and quantum) algorithms to compute the edit distance are also nearly quadratic. As mentioned earlier, the result by \cite{EditDist-BackursIndyk-15} proves that these near quadratic time bounds might be tight. They show that a sub-quadratic classical algorithm for computing the $\editdist$ problem would imply that $\SETH$ (refer to Conjecture~\ref{conj:SETH}) is false.
$\SETH$ also implies quadratic lower bounds for many other string comparison problems like $\LCS$, $\DTW$ whose trivial algorithms are also based on dynamic programming \cite{EditDist-BackursIndyk-15, QuadLowerBounds-BringmannKunnemann-15}. Bouroujeni \emph{et al.}~in \cite{approxQEdit-BoroujeniEtAl-18} give a sub-quadratic quantum algorithm for approximating edit distance within a constant factor which was followed by a better classical algorithm in \cite{approxCEdit-ChakrabortyEtAl-18} by Chakraborty \emph{et al.}
However, no quantum improvements over the classical algorithms in the exact case are known to the best of our knowledge. 
Investigating why this is the case is an interesting open problem, which can be addressed in two directions. We formulate the following questions for the example of the $\editdist$ problem.

\begin{enumerate}
    \item Is there a bounded-error quantum algorithm for $\editdist$ that runs in a sub-quadratic amount of time?
    \item Can we use a different reduction to raise the linear lower bound for $\editdist$ that we achieve under $\DCQSETH$?
\end{enumerate}

While the first question still remains open, we address the second question in this section. Independently from our results, Ambainis \emph{et al.}~\cite{Dyck-AmbainisEtAl-19} present a quantum query lower bound of $\Omega(n^{1.5-o(1)})$ for the $\editdist$ problem when solved using the most natural approach by reducing $\editdist$ to connectivity on a 2D grid.
However, that doesn't rule out the possibility of other $\widetilde{O}(n^{1.5-\alpha})$ quantum algorithms for the $\editdist$ problem, for $\alpha > 0$.
In this section, using (a promise version of) the $\NCQSETH$ conjecture we prove a conditional quantum time lower bound of $\Omega(n^{1.5})$ for the $\editdist$ problem.

\subsection{The Edit Distance problem and the Alignment Framework}
\label{sec:EditDistAlignFrame}

Formally, the $\editdist$ problem is defined as follows:
\begin{definition}[The $\editdist$ problem]
\label{def:EditDistance}
Given two strings $a$ and $b$ over an alphabet set~$\Sigma$, the edit distance between $a$ and $b$ is the minimum number of operations (insertions, deletions, substitutions) on the symbols required to transform string $a$ to $b$ (or vice versa).
\end{definition}

One way to visualize the $\editdist$ problem is by using the alignment framework, also mentioned in  \cite{QuadLowerBounds-BringmannKunnemann-15}.  Let $a$, $b$ be two strings, of length $n$ and $m$ respectively, and let $n \geq m$. An \textit{alignment} is a set $A=\{(i_1,j_1),\dots,(i_k,j_k)\}$ with $0 \leq k \leq m$ such that $1 \leq i_1 <\dots<i_k \leq n$ and $1 \leq j_1 < \dots <j_k \leq m$. The set $\mathcal{A}_{n,m}$ denotes the set of all alignments over the index sets $[n]$ and $[m]$. We now claim the following:
\begin{fact}
Let $a$,$b$ be strings of length $n,m$, respectively. The edit distance between $a$ and $b$ is
\begin{equation*}
    \delta(a,b)=\min_{A \in \mathcal{A}_{n,m}} \alignCost(A),
\end{equation*}
where $n=|a|, m=|b|$, $\alignCost(A)=\sum_{(i,j) \in A} \delta(a[i],b[j]) + n +m -2|A|$, and $a[i]$ denote the $i^{th}$ symbol of string $a$, while $b[j]$ denotes the $j^{th}$ symbol of string $b$.
\end{fact}
\begin{proof}
There are many ways to transform the string $a$ into the string $b$ and each alignment $A \in \mathcal{A}_{n,m}$ specifies one such way.
For any alignment $A \in \mathcal{A}_{n,m}$ the $\alignCost(A)$ denotes the number of operations (insertions, deletions, substitutions) required to transform $a$ to $b$ under the alignment $A$.
As edit distance is defined to be the \emph{minimum} number of operations required to transform $a$ to $b$, we minimize the $\alignCost$ over all the alignments in $\mathcal{A}_{n,m}$ to get the edit distance $\delta(a,b)$.
\end{proof}
We chose the alignment framework to visualize the $\editdist$ problem because in this framework the edit distance between two strings can be related to the sum of edit distance between pairs of \emph{some} symbols from these two strings, a recursive behaviour that we will extensively use in the following results.

\subsection{Reduction from \texorpdfstring{$\text{BP-PP}_{\text{edit}}$}{BP-PPedit} to the Edit Distance problem}
\label{sec:BPPedittoEditDistance}
We present a conditional quantum time lower bound for the $\editdist$ problem as one of the first consequences of our $\NCQSETH$. First we define a promise property $\propertyPPedit$. We then give an efficient reduction from the problem of computing the property $\propertyPPedit$ on truth tables of some\footnote{The mention of \emph{some} is important because of two reasons: (1) The property $\text{PP}_{\text{edit}}$ is a promise property defined on truth tables of some branching programs. (2) As a part of the reduction for a given branching program as an input we construct gadgets whose sizes depend on the size of the input, and to avoid having the length of the gadgets be too large, we restrict ourselves to branching programs of size $2^{o(\sqrt{n})}$.} non-deterministic branching programs \cite{CCBook-AroraBarak-09} to the $\editdist$ problem. We provide a quantum time lower bound for computing the property $\propertyPPedit$ on a set of non-deterministic branching programs using the promise version of $\NCQSETH$, which because of the reduction translates to a conditional quantum time lower bound of $\Omega(n^{1.5})$ for the $\editdist$ problem.

The first part of our reduction mimics the approach in the \cite{EditDist-AbboudEtAl-16} paper.
Given a non-deterministic branching program $S$ (Definition~\ref{def:branchingPrograms} in Appendix~\ref{sec:AppBranchingPrograms}) with $n$~input variables, we do the following: Let $X_1=\{x_1, x_2,\dots,x_{{n/2}}\}$ and $X_2=\{x_{{n/2}+1}, x_{{n/2}+2},\dots,x_{n}\}$ be the first and the last half of the input variables to $S$, respectively.
Let $A=(a_1,a_2,\dots,a_{2^{n/2}})$ and $B=(b_1,b_2,\dots,b_{2^{n/2}})$ be two sequences containing all the elements from the set $\{0,1\}^{n/2}$ in the lexicographical order such that every pair $(a,b) \in A \times B$ together forms an input to $S$.
For each set $A$ and $B$, the reduction constructs two long sequences $x$ and $y$, such that these sequences are composed of subsequences (also referred to as \textit{gadgets}) that correspond to elements of $A$ and $B$, respectively.
We observe that computing the edit distance between these sequences $x$ and $y$ is equivalent to computing the property $\propertyPPedit$ of the truth table of the branching program $S$.
Therefore, we establish a connection between the $\text{BP-PP}_{\text{edit}}$ problem (Definition~\ref{def:BP-P}) which is solvable in exponential time and the $\editdist$ problem that is solvable in polynomial time.

\begin{definition}[The $\propertyPPedit$ property]
\label{def:PP_EditProperty}
Let $\matrixM$ be a boolean matrix of size $K \times L$ where $\matrixM_{ij}=\{0,1\}$ denotes the entry in the $i^{th}$ row and the $j^{th}$ column. 
We define a \textbf{path} $\pathR=((i_1,j_1), (i_2,j_2),\dots,(i_k,j_k))$ as a sequence of positions in the matrix $\matrixM$ which satisfy the following conditions:

\begin{enumerate}
    \item The column indices in a path are ordered, i.e.,\ $1= j_1 \leq j_2 \leq \dots\leq j_k=L$. This ensures that the path can only start from a cell in the first column and must end in a cell in the last column and the path progresses from left to right in the matrix.
    \item For all $p\in [k-1]$, either $i_{p+1}= i_{p}$ or $i_{p+1}>i_{p}$ or $i_{p+1}<i_{p}$. If $i_{p+1}=i_{p}$ then $j_{p+1}=j_p+1$. However, if $i_{p+1} \ne i_{p}$ then we say there is a jump to another row at $(i_p,j_p)$. When jumping to a row above, i.e.\ when $i_{p+1} < i_p$, then $j_{p+1}=j_p + i_p-i_{p+1}$. Whereas, while jumping to a row below, i.e.\ when $i_{p+1} > i_p$, then $j_{p+1}=j_p$.
    \item Finally, $\forall p$ such that $1<p<k$, if $i_p \neq i_{p-1}$ then $i_{p+1} = i_p$. 
    
\end{enumerate}
Let $\setOfPaths_{K,L}$ be a set of all possible \emph{paths} for a matrix of size $K \times L$. The \emph{cost} associated with a path $\pathR$ for a given matrix $\matrixM$ depends on the entries of $\matrixM$ and is defined as:
\begin{equation*}
    \pathCost(\matrixM, \pathR, \mu)=\sum_{(i_p,j_p) \in \pathR} C_{\matrixM_{i_pj_p}} + \sum_{(i_p,j_p) \in \pathR, i_p \neq i_{p+1}} \underbrace{(C_{jump}|i_{p+1}-i_p| + \mu - C_{\matrixM_{i_pj_p}}-C_{\matrixM_{i_{p+1}j_{p+1}}} }_{\textit{Jump costs}}),
\end{equation*}
where $C_0$, $C_1$ with ($C_1 < C_0$) are some fixed constants, $C_{jump}$ depends on values of $T$ and $S_G$\footnote{The constants $C_0=Q$ and $C_1=Q-\rho$ where as $C_{jump}=2T+S_G$ where $S_G$ is a constant but the variable $T$ is not. The details about the constants $Q,\rho, S_G$ and the variable $T$ are mentioned in Theorem \ref{Thm:BPSATtoEDIT}.} and $\mu$ is an integer parameter between $[0,Q]$ for some constant $Q$.
We define:
\begin{equation*}
    \minPathCost(\matrixM, \mu) = \min_{\pathR \in \setOfPaths_{K,L}}\pathCost(\matrixM, \pathR, \mu).
\end{equation*}

Given a fixed threshold value $T_r$\footnote{We fix $T_r=\frac{3L}{4}C_0+\frac{L}{4}C_1$.}, the property $\propPedit: \{0,1\}^K \times \{0,1\}^L \times \{0,1,\dots ,Q\} \rightarrow \{0,1\}$ is defined as follows:
\begin{equation*}
    \propPedit(\matrixM, \mu)=\begin{cases}
    1, & \text{if } (\minPathCost(\matrixM, \mu) < T_r),\\ 
    0 & \text{if } (\minPathCost(\matrixM, \mu) \geq T_r).
  \end{cases}
\end{equation*}
We now define a promise version of the $\propPedit$ property, namely $\propertyPPedit: \{0,1\}^K \times \{0,1\}^L \rightarrow \{0,1\}$ as follows:
\begin{equation*}
    \propertyPPedit(\matrixM)=\begin{cases}
    1, & \text{if } (\propPedit(\matrixM, \mu=Q)=1),\\
    0, & \text{if } (\propPedit(\matrixM, \mu=0)=0).
  \end{cases}
\end{equation*}
\end{definition}

\begin{definition}[$\text{BP-PP}_{\text{edit}}$ problem]
\label{def:BP-P}
Given a non-deterministic branching program $S$ with $n$ input variables, decide if $\propertyPPedit(\matrixM^{\truthtable(S)})=1$. Here $\truthtable(S)$ denotes the truth table of the function computed by the branching program $S$ and $\matrixM^{\truthtable(S)}$ denotes the $\encodingName$\footnote{$\encodingName$: Let the truth table $\truthtable(S)=X_1X_2 \dots X_{2^n}$ be of length $2^n$. Then the matrix $\matrixM^{\truthtable(S)}$ of size $(2^{{n/2}+1}-1) \times 2^{{n/2}}$ (refer to Figure~\ref{fig:MatrixM}) is generated in the following way: 
\begin{equation*}
    \matrixM^{\truthtable(S)}_{ij}=\begin{cases}
    X_{2^{{n/2}}(i+j-2^{{n/2}})+j}, & \text{if } (0 < (i+j-2^{{n/2}}) \leq 2^{{n/2}}),\\
    0, &\text{otherwise.}
    \end{cases}
\end{equation*}
Here $\matrixM^{\truthtable(S)}_{ij}$ denotes the entry at the $i^{th}$ row and the $j^{th}$ column of the matrix $\matrixM^{\truthtable(S)}$.} of $\truthtable(S)$.
\end{definition}

We now provide the reduction from the $\text{BP-PP}_{\text{edit}}$ promise problem to the $\editdist$ problem. The main lemmas and facts pertaining to the reduction are mentioned in this section while the detailed proofs are given in Appendix \ref{sec:AppReduction}.

\begin{definition}[The set $\mathcal{S}$]
\label{def:mathcalS}
Given a branching program $S$, let $\desc(S)$ denote a standard encoding of $S$ as a binary string. The set $\mathcal{S}$ denotes the set of non-deterministic branching programs with $n$ input variables such that:
\begin{equation*}
    \mathcal{S}=\{S \text{ }| \text{ } \matrixM^{\truthtable(S)} \in \propertyPPedit^{-1}(0) \cup \propertyPPedit^{-1}(1)\}.
\end{equation*}
 
\end{definition}

Having defined the $\text{BP-PP}_{\text{edit}}$ problem and the set $\mathcal{S}$, we now present the main theorem of this section.

\begin{theorem}\label{Thm:BPSATtoEDIT}
There is a reduction from the $\text{BP-PP}_{\text{edit}}$ problem on non-deterministic branching programs of size $2^{o(\sqrt{n})}$ (length $Z$ and width $W$) from set $\mathcal{S}$ to an instance of the $\editdist$ problem on two sequences of length $N = Z^{O(\log W)}2^{n/2}$, and the reduction runs in $O(N)$ time.
\end{theorem}
\begin{proof}
The reduction is as follows: Let $S \in \mathcal{S}$ be a branching program with $n$ input variables of size $2^{o(\sqrt{n})}$. Let $X_1=\{x_1, x_2, \dots, x_{{n/2}}\}$ and $X_2=\{x_{{n/2}+1}, x_{{n/2}+2}, \dots, x_n\}$ be the inputs to the branching program $S$. Let $A = (a_1, a_2, \dots,a_{2^{n/2}})$ and $B=(b_1, b_2,\dots,b_{2^{n/2}})$ such that both the sequences contain all the elements from the set $\{0,1\}^{n/2}$ in the lexicographical order. Then construct \textit{gadget sequences} such that the following conditions are met:
\begin{enumerate}
    \item Using the construction mentioned by~\cite{EditDist-AbboudEtAl-16}, construct the \textit{gadget}~$G(a)$ (of length~$Z^{O(\log W)}$, using symbols from $\Sigma=\{0,1\}$) for each $a \in A$ and another \textit{gadget}~$\overline{G}(b)$ (also of length $Z^{O(\log W)}$ using symbols from $\Sigma'=\{0,1\}$) for each $b \in B$, such that $\forall (a,b) \in A \times B$, $\delta(G(a),\overline{G}(b))= Q'-\rho$ if the pair $(a,b)$ satisfies the branching program $S$ and $\delta(G(a),\overline{G}(b))= Q'$ otherwise, for some constants $Q'$ and $\rho$ $\in \mathbb{Z}$. Recall that $\delta(G(\cdot),\overline{G}(\cdot))$ refers to the edit distance between the gadgets $G(\cdot)$ and $\overline{G}(\cdot)$. We slightly modify the construction of gadgets to match the sizes of these gadgets $|G(\cdot)|=|\overline{G}(\cdot)|$ denoted by $S_G$. Refer to the Theorem \ref{thm:AppGadgetSizePreserved} in Appendix \ref{sec:AppReduction} for a simple proof on how we achieve this for another constant $Q$, same $\rho$ and a bigger alphabet set $\Sigma=\{0,1,2\}$.
    \item We set $T=\omega(\log N)$ where $N=Z^{O(\log W)}2^{n/2}$ and construct the final sequences in the following way:
    \begin{equation*}
    \label{eq:AlignmentGadgetX1}
        x:= (\bigcirc_{i=1}^{|A|-1}5^Tr6^T)(\bigcirc_{a \in A}5^TG(a)6^T)(\bigcirc_{i=1}^{|A|-1}5^Tr6^T)
    \end{equation*} 
    \begin{equation*}
    \label{eq:AlignmentGadgetY1}
       y:=7^{|x|}(\bigcirc_{b \in B}5^T\overline{G}(b)6^T)7^{|x|}
    \end{equation*} 
    Here $r$ is a dummy gadget sequence such that $\delta(r,\overline{G}(b))= Q$, $\forall b \in B$ \cite{EditDist-AbboudEtAl-16} and $5^T$ (or $6^T$) known as a \emph{separator} represents the symbol $5$ (or $6$) occurring $T$ times. The choice of $T$ is made in a way that in an \emph{optimal alignment}\footnote{We say an alignment is an optimal alignment if the cost of the alignment is the edit distance.} the separators from $x$ align with the separators in $y$, hence forcing the gadgets to align with other gadgets.
\end{enumerate}  

In Section \ref{sec:EditDistAlignFrame} we saw how edit distance between two strings can be associated with minimum $\textit{alignment-cost}$ when viewed as an alignment of symbols.
Unfortunately, we cannot use the alignment framework in the same way at a gadget level.
Therefore, instead of using the $\textit{alignment}$ framework, which by definition is a set of pairs such that pairing is between two indices, we define a variant, which we call the $\textit{coarse alignment}$, which is a sequence of pairs, but the pairing here can be between an index and a sequence of indices (or vice-versa).
We observe that the edit distance $\delta(x,y)$ can be expressed as the cost of an optimal \emph{coarse alignment} of gadgets in $x$ and $y$ as presented in Lemma~\ref{thm:EditDistanceCourseAlignment}.

\begin{definition}[Coarse alignment]
\label{def:course_alignment}

Let $n',n'' \in [n]$ and $n'<n''$, we say $\mathcal{I}_{n',n''}=\{n',n'+1,\dots,n''\}$ and let $\mathcal{J}_{m}=[m]$ be two sets of indices. \\A \textit{coarse alignment} $A$ is defined as a sequence $((p_1,q_1),(p_2,q_2),\dots,(p_k,q_k))$, such that:
\begin{enumerate}
    \item $\forall i \in [k], p_i$ are sequences and $\bigcup_{i=1}^{k}p_i=\mathcal{I}_{n',n''}$. Similarly, $\forall i \in [k], q_i$ are also sequences and $\bigcup_{i=1}^{k}q_i=\mathcal{J}_{m}$.
    \item $\forall i \in [k], p_i\neq \emptyset$ and $\forall i \in [k], q_i\neq \emptyset$. 
    \item $\forall i,j \in [k],$ if $i\neq j, p_i \cap p_j=\emptyset$. Similarly, $q_i \cap q_j=\emptyset$, whenever $i\neq j$.
    \item $\forall i,j \in [k],$ $\forall u \in p_i$ and $\forall v \in p_j$ , $u < v$ if $i < j$. Similarly, $\forall u \in q_i$ and $\forall v \in q_j$ , $u < v$ whenever $i < j$.
    \item $\forall i \in [k]$, $|p_i|=1$ or $|q_i|= 1$.
\end{enumerate}
Let the set $\mathcal{A}_{n',n'',m}$ denote the set of coarse alignments given the indices sets $\mathcal{I}_{n',n''}$ and $\mathcal{J}_{m}$. We define the set $\mathcal{C}_{n,m}=\bigcup_{i,j \in [n], i<j}\mathcal{A}_{i,j,m}$ to denote the set of all possible coarse alignments given $n$ and $m$.
\end{definition}

The structure of the sequences $x$ and $y$, i.e., the padding of $5$s and $6$s between the gadgets and the padding of the $7$s in the beginning and end of the sequence $y$, ensures that the edit distance $\delta(x,y)$ has the following peculiar behaviour:

\begin{lemma}[Lemma \ref{thm:AppEditDistanceCourseAlignment} in Appendix \ref{sec:AppReduction}]
\label{thm:EditDistanceCourseAlignment}
There exists a coarse alignment $A \in \mathcal{C}_{(3 \cdot 2^{n/2}-2), 2^{n/2}}$, such that the edit distance between the two sequences $x$ and $y$ is:
\begin{equation*}
    \delta(x,y) = 2|x|+ \min_{A \in \mathcal{C}_{(3 \cdot 2^{n/2}-2), 2^{n/2}}} \editCost(A),
\end{equation*}
where $\editCost(A)=\sum_{(i,j)\in A}\delta(u_i,v_j)$ such that $u_i=\bigcirc_{p \in i}5^T g_p6^T$ and $v_j=\bigcirc_{q \in j}5^T \overline{G}(b_q)6^T$. Also $g_p=G(a_{p-2^{n/2}})$ when $0 < (p-2^{n/2}) \leq 2^{n/2}$ and $g_p=r$ otherwise. Here $r$ denotes the dummy gadget.
\end{lemma}

The proof of Lemma \ref{thm:EditDistanceCourseAlignment} uses results of Lemmas \ref{thm:Appx2Format}, \ref{thm:AppSeparatorAligns}, Fact \ref{thm:AppFactdiffAlignment} and Corollary \ref{thm:AppCorJustX2} mentioned in Appendix \ref{sec:AppReduction} where we also provide insights on how edit distance between sequences like $x$ and $y$ behave. We now prove the correctness of our reduction. 

\begin{lemma}[Lemma \ref{thm:AppCorrectnessOfReduction} in Appendix \ref{sec:AppReduction}]
\label{thm:CorrectnessOfReduction}
For every $n$, there exists a constant $C^* \in \mathbb{Z}$ such that
\begin{equation*}
    \delta(x,y) < C^*
\end{equation*}
if and only if $\propertyPPedit(\matrixM^{\truthtable(S)})=1$.
\end{lemma}     
\begin{proof}
If $\propertyPPedit(\matrixM^{\truthtable(S)})=1$ then it implies $\propPedit(\matrixM^{\truthtable(S)},Q)=1$ because of the promise that the branching program $S$ belongs to the set $\mathcal{S}$. The statement $\propPedit(\matrixM^{\truthtable(S)},Q)=1$ implies that there exists a $path$ $\pathP \in \setOfPaths_{2L-1,L}$ (here $L=2^{n/2}$) such that even with the highest allowed jump parameter $\mu=Q$ the path-cost is $\pathCost(\matrixM^{\truthtable(S)},\pathP, Q) < T_r$, for a fixed threshold $T_r$. Which invariably means that $\forall \mu$, $\pathCost(\matrixM^{\truthtable(S)},\pathP, \mu) < T_r$. Using the Algorithm \ref{Alg:PtoCmain} that generates a  $\text{coarse alignment}$ $\coarseAlignC \in \mathcal{C}_{3L-2,L}$ for a given path $\pathP \in \setOfPaths_{2L-1,L}$ as an input, we get a coarse alignment $\coarseAlignC$ corresponding to path $\pathP$. We get $\editCost(\coarseAlignC)< T_r$ because $\forall \mu, \pathCost(\matrixM^{\truthtable(S)},\pathP, \mu) < T_r$.\footnote{This happens due to the choices of $C_0$, $C_1$, $C_{jump}$ that were made in Definition \ref{def:PP_EditProperty}.} Also, from Lemma~\ref{thm:EditDistanceCourseAlignment} we have that $\delta(x,y)=2|x| + \min_{\coarseAlignC \in \mathcal{C}_{3 L-2, L}} \editCost(\coarseAlignC)$. This implies $\delta(x,y)< 2|x| + T_r$ and, we set our constant $C^*=2|x|+T_r$. 

We now prove the other direction. If $\propertyPPedit(\matrixM^{\truthtable(S)})=0$ it implies that $\propPedit(\matrixM^{\truthtable(S)},0)=0$ as the branching program $S\in \mathcal{S}$. Which in turn implies that $\forall \mu, \forall \pathP \in \setOfPaths_{2L-1,L}$, $\pathCost(\matrixM^{\truthtable(S)},\pathP,\mu)\geq T_r$. Using the result from Lemma \ref{thm:AppPtoCmain} and \ref{thm:AppEitherPathorD} in Appendix \ref{sec:AppReduction} we show that if $\forall \mu, \forall \pathP \in \setOfPaths_{2L-1,L}$, $\pathCost(\matrixM^{\truthtable(S)},\pathP,\mu)\geq T_r$ then $\forall \coarseAlignC \in \mathcal{C}_{3L-2, L}$, $\editCost(\coarseAlignC) \geq T_r$ which implies $\delta(x,y)\geq 2|x| + T_r$. Thus, implying $\delta(x,y) \geq C^*$.
\end{proof}

This constitutes the proof of the reduction from the $\text{BP-PP}_{\text{edit}}$ problem on branching programs of size $2^{o(\sqrt{n})}$ from the set $\mathcal{S}$ to the $\editdist$ problem. 
\end{proof}

\subsection{The Quantum Time Lower Bound for the Edit Distance problem}
In the previous sub-section we gave a reduction from the $\text{BP-PP}_{\text{edit}}$ problem on branching programs of size $2^{o(\sqrt{n})}$ from set $\mathcal{S}$ to the $\editdist$ problem. Therefore, if we prove that the time taken to compute the $\propertyPPedit$ on these branching programs in the white-box setting in $\epsilon$-bounded error model is $\Omega(2^{0.75n})$ then because of the reduction we prove a quantum time lower bound of $\Omega(n^{1.5})$ for the $\editdist$ problem. 

To achieve the quantum time lower bound for the $\text{BP-PP}_{\text{edit}}$ problem we use the results of the Theorem \ref{thm:PPeditQuery} (query complexity of $\propertyPPedit$) and Conjectures~\ref{thm:ConjPromiseNCQSETH} (promise version of $\NCQSETH$) and~\ref{thm:ConjPPeditNCoblivious} ($\propertyPPedit \in \compressionOblivious(\NC \cap \mathcal{S})$) mentioned below. 

\begin{theorem}[Theorem \ref{thm:AppPeditQuery}, Corollary \ref{thm:AppCorQofPPedit} in Appendix \ref{sec:AppQueryComplexity}]
\label{thm:PPeditQuery}
The bounded-error quantum query complexity for computing the property $\propertyPPedit$ on matrices that are $\encodingName$ of truth tables of non-deterministic branching programs with $n$ input variables of size $2^{o(\sqrt{n})}$ from set $\mathcal{S}$ is $\Omega(2^{0.75n})$. \end{theorem}

\begin{conjecture}[Promise version of $\NCQSETH$] 
\label{thm:ConjPromiseNCQSETH}
For the class of representations $\NC$, i.e., the set of poly sized circuits of polylogarithmic depth consisting of fan-in 2 gates, for all properties $\propertyP \in \compressionOblivious(\NC \cap \mathcal{S})$, we have $\qTimeWB(\propertyP|_{\NC \cap \mathcal{S}}) \geq \Omega(\Query(\propertyP|_{\mathcal{S}}))$.
\end{conjecture}

\begin{conjecture}
\label{thm:ConjPPeditNCoblivious}
The property $\propertyPPedit$ is $\NC \cap \mathcal{S})$-compression oblivious.
\end{conjecture}

\begin{theorem}
\label{thm:PPeditTime}
The $\epsilon$-bounded error quantum time complexity for computing the property $\propertyPPedit$ in the white-box setting is $\qTimeWB(\propertyPPedit|_{\mathcal{\NC \cap S}})=\Omega(2^{0.75n})$ under a promise version of $\NCQSETH$.
\end{theorem}
\begin{proof}
Combining the results of  Theorem \ref{thm:PPeditQuery}, Conjectures \ref{thm:ConjPromiseNCQSETH} and \ref{thm:ConjPPeditNCoblivious} we get $\qTimeWB(\propertyPPedit|_{\NC \cap \mathcal{S}}) = \Omega(\Query(\propertyPPedit|_{\mathcal{S}}))=\Omega(2^{0.75n})$. 
\end{proof}

\begin{theorem}
\label{thm:editDistTime}
Assuming Conjecture~\ref{thm:ConjPromiseNCQSETH}, the promise version of $\NCQSETH$, and assuming $\propertyPPedit \in \compressionOblivious(\NC \cap \mathcal{S})$, the bounded-error quantum time complexity for computing the \\$\editdist$ problem is $\Omega(n^{1.5})$.
\end{theorem}
\begin{proof}
Using the reduction from the $\propertyPPedit$ problem on branching programs with $n$ input variables of size $2^{o(\sqrt{n})}$ from set $\mathcal{S}$ to the $\editdist$ problem (Theorem \ref{Thm:BPSATtoEDIT}) and using the results from Theorem~\ref{thm:PPeditTime} we obtain the conditional quantum time lower bound of $\Omega(n^{1.5})$ for the $\editdist$ problem. 
\end{proof}

\subsection{Lower bound for the restricted Dyck language}
\label{sec:Dyck}
Consider the restricted Dyck language, the language of balanced parentheses of depth bounded by~$k$.
The study of the quantum query complexity of this language was initiated in \cite{Trichotomy-AaronsonEtAl-19} by Aaronson, Grier, and Schaeffer where they provide an $\widetilde{O}(\sqrt{n})$ algorithm to decide the language for a constant~$k$. As a corollary to Theorem~\ref{thm:AppPeditQuery} in Appendix~\ref{sec:AppQueryComplexity} we show that the query complexity of restricted Dyck language is linear for any $k=\omega(\log n)$, partially answering an open question posed by the authors in \cite{Trichotomy-AaronsonEtAl-19}.
\begin{theorem}
The quantum query complexity of the restricted Dyck language is $\Omega(n^{1-o(1)})$ for any $k=\omega(\log n)$.
\end{theorem}
\begin{proof}[Proof (sketch)]
In the proof of Theorem~\ref{thm:AppPeditQuery} in Appendix~\ref{sec:AppQueryComplexity} we have constructed sets of length $n$ 0/1 strings such that, for each string in the sets, every prefix is at most some distance $d$ away from balanced, and such that the query complexity of deciding whether these strings are precisely balanced or not is $\Omega(n)$, whenever $d=\omega(\log n)$.

We can directly use the rows of the matrices that form the adversary set in Theorem~\ref{thm:AppPeditQuery} and use this as the adversary set to lower bound the query complexity of the restricted Dyck language with $k=2d$ in the following way:
\begin{enumerate}
    \item Interpret 0 as an open bracket `(', and 1 as a closing bracket `)'.
    \item For any row $r$ from the matrices of the adversarial set, we define
    \begin{equation*}
        r' = 0^d r 1^d
    \end{equation*}
\end{enumerate}
If $r$ is balanced, and satisfies the promise of never being $d$-far from balanced, this is a valid 1-instance for the restricted Dyck language with $k=2d$.
Additionally, if $r$ is not balanced, then this is a valid 0-instance of the restricted Dyck language.
Therefore, the adversary bound also is valid for this problem (losing an additive $2d$ in the lower bound, which is negligible in the parameter range that needs to be considered).
\end{proof}

\section{Conclusion and Future Directions}
\label{sec:Conclusion}

We presented a quantum version of the strong exponential-time hypothesis, as QSETH, and demonstrated several consequences from QSETH.
These included the transfer of previous Orthogonal-Vector based lower bounds to the quantum case, with a quadratically lower time bound than the equivalent classical lower bounds.
We also showed two situations where the new QSETH does not lose this quadratic factor: 
a lower bound showing that computing edit distance takes time $n^{1.5}$ for a quantum algorithm, and an $n^2$ quantum lower bound for the Proofs of Useful Work of Ball, Rosen, Sabin and Vasudevan~\cite{uPoW-BallRosenSabinVasudevan-17}, both conditioned on QSETH.

Possible future applications for the QSETH framework are numerous.
Most importantly, the QSETH can potentially be a powerful tool to prove conditional lower bounds for additional problems in $\BQP$.
The most natural first candidates are other string problems, such as computing the $\LCS$ for example, but there are many other problems for which the `basic QSETH' does not immediately give tight bounds.

Additionally, the notion of \emph{compression oblivious} properties are potentially interesting as an independent object of study.
We expect most natural properties to be compression oblivious, but leave as an open question what complexity-theoretic assumptions are needed to show that, e.g., the parity function is compression oblivious.

Future directions also include a careful study of quantum time complexity of the other core problems in fine-grained complexity, such as $\threeSUM$ and $\APSP$.
Just like with satisfiability, the basic versions of these problems are amenable to a Grover-based quadratic speedup.
It is possible that extensions of those key problems can be used to prove stronger conditional lower bounds, in a similar way to the reduction that was used for $\editdist$ in the current work.

\section*{Acknowledgments}

We would like to thank Andris Ambainis, Gilles Brassard, Fr\'ed\'eric Magniez, Miklos Santha, and Ronald de Wolf for helpful discussions.

SP is supported by the Robert Bosch Stiftung. HB, SP, and FS are additionally supported by NWO Gravitation grants NETWORKS and QSC, and EU grant QuantAlgo.


\bibliographystyle{alpha}
\bibliography{qseth.bib}

\begin{appendices}

\section{Branching Programs}
\label{sec:AppBranchingPrograms}

\begin{definition}[Non-deterministic Branching Programs]
\label{def:branchingPrograms}
A non-deterministic branching program is a directed acyclic graph with $n$ input variables~$x_1, \dots ,x_n$. It is a $Z$ layered directed graph with each layer having a maximum of $W$ nodes and the edges can only exists between nodes of neighbouring layers $L_i$ and $L_{i+1}$, $\forall i \in [Z-1]$. Every edge is labelled with a constraint of the form $(x_i=b)$ where $x_i$ is an input variable and $b \in \{0,1\}$. One of the nodes in the first layer is marked as the start node, and one of the nodes in the last layer is marked as the accept node. An evaluation of a branching program on an input $x_1, \dots ,x_n$ is a path that starts at the start node and non-deterministically follows an edge out of the current node. The branching program accepts the input if and only if the path ends up in the accept node. The size of this non-deterministic branching program is the total number of edges i.e. $O(W^2Z)$.
\end{definition}
With a non-deterministic branching program $S$ on $n$ inputs, we associate the boolean function $f=[S]$ as the function computed by the branching program $S$. We use $\truthtable(S)$ to denote the truth table of the function $f$ computed by the branching program $S$ and, a standard encoding of $S$ as a binary string is denoted by $\desc(S)$.  


\section{Theorems related to reduction from the \texorpdfstring{$\text{BP-PP}_{\text{edit}}$}{BP-PPedit} problem to the Edit Distance problem}
\label{sec:AppReduction}

\begin{theorem}
\label{thm:AppGadgetSizePreserved}
Given two strings $a,b \in \{0,1\}^*$ with $|a|>|b|$ such that, either $\delta(a,b)=Q'$ or $\delta(a,b)=Q'-\rho$ for some constants $Q', \rho \in \mathbb{Z^+}$ and $\rho < Q'$, we can create strings $a_{new},b_{new} \in \{0,1,2\}^*$ such that,

\begin{equation*}
\begin{split}
    a_{new}=2^{|a|}\bigcirc a, \\
    b_{new}=0^{|a|-|b|}2^{|a|}\bigcirc b,
\end{split}
\end{equation*}
and,
\begin{equation*}
\begin{split}
    \delta(a_{new},b_{new})=Q-\rho \text{ iff } \delta(a,b)=Q'-\rho, \\
    \delta(a_{new},b_{new})=Q \text{ iff } \delta(a,b)=Q',
\end{split}
\end{equation*}
for another constant $Q \in \mathbb{Z^+}$ and $\rho < Q$ and $|a_{new}|=|b_{new}|$.
\end{theorem}
\begin{proof}
It follows from the construction of $a_{new},b_{new}$ that $|a_{new}|=|b_{new}|$. It is also easy to see that $\delta(a_{new},b_{new}) \leq |a|-|b|+\delta(a,b)$. We will now prove that $\delta(a_{new},b_{new}) =|a|-|b|+\delta(a,b)$ for any such $a,b$. Consider the last $2$ of the subsequence $2^{|a|}$ in the string $a_{new}$. This symbol $2$ could either get deleted, or get matched or get substituted. If this $2$ gets substituted with a symbol from the substring $b$ in $b_{new}$ then there is a symbol $2$ in $b_{new}$ that has to be inserted, which is not an optimal thing to do. If this symbol $2$ from $a_{new}$ gets substituted with a symbol in $0^{|a|-|b|}$ then there is a symbol $2$ in $b_{new}$ that has to be substituted with a symbol from the substring $a$ in $a_{new}$ or has to be inserted, either ways not an optimal thing to do. A similar argument holds for the case where this $2$ from $a_{new}$ gets deleted. Therefore, the only option left is that it gets matched with a $2$ in $b_{new}$. By repeating this argument for all the $2$s in $a_{new}$ we can say that in an optimal alignment all the $2$s from $a_{new}$ will align with all the $2$s from $b_{new}$. Therefore, $\delta(a_{new},b_{new}) =|a|-|b|+\delta(a,b)$ and $Q=Q'+|a|-|b|$.
\end{proof}

\begin{lemma}
\label{thm:AppEditDistanceCourseAlignment}
There exists a coarse alignment $\coarseAlignC \in \mathcal{C}_{(3 \cdot 2^{n/2}-2), 2^{n/2}}$, such that the edit distance between the two sequences $x$ and $y$ is:
\begin{equation*}
    \delta(x,y) = 2|x|+ \min_{\coarseAlignC \in \mathcal{C}_{(3 \cdot 2^{n/2}-2), 2^{n/2}}} \editCost(\coarseAlignC),
\end{equation*}
where $\editCost(\coarseAlignC)=\sum_{(i,j)\in \coarseAlignC}\delta(u_i,v_j)$ such that $u_i=\bigcirc_{p \in i}5^T g_p6^T$ and $v_j=\bigcirc_{q \in j}5^T \overline{G}(b_q)6^T$. Also $g_p=G(a_{p-2^{n/2}})$ when $0 < (p-2^{n/2}) \leq 2^{n/2}$ and $g_p=r$ otherwise. Here $r$ denotes the dummy gadget.
\end{lemma}
\begin{proof}
We are given two sequences $x$ and $y$, such that
\begin{center}
    \begin{equation}
     \label{eq:AppAlignmentGadgetsXY}
     \begin{split}
     x:= (\bigcirc_{i=1}^{|A|-1}5^Tr6^T)(\bigcirc_{a \in A}5^TG(a)6^T)(\bigcirc_{i=1}^{|A|-1}5^Tr6^T),\\
     y:=\underbrace{7^{|x| }}_{y_1}\underbrace{(\bigcirc_{b \in B}5^T \overline{G}(b)6^T)}_{y_2}\underbrace{7^{|x|}}_{y_3}.
     \end{split}
    \end{equation}
\end{center}

Recall that $A = (a_1, a_2, \dots ,a_{2^{n/2}})$ and $B=(b_1, b_2, \dots ,b_{2^{n/2}})$ and both the sequences contain all the elements from the set $\{0,1\}^{n/2}$ in the lexicographical order. The gadgets $r, G(a)$ and $\overline{G}(b) \in \{0,1,2\}^{*}$ and $T>>|G(\cdot)|=|\overline{G}(\cdot)|=|r|=S_G$.

\begin{fact}[Fact 5.7 in \cite{QuadLowerBounds-BringmannKunnemann-15}]
\label{thm:AppFactdiffAlignment}
Let $y_1=7^{|x|}$, $y_2=\bigcirc_{b \in B}5^T\overline{G}(b)6^T$, $y_3=7^{|x|}$ as mentioned in Equation~\ref{eq:AppAlignmentGadgetsXY} above. We have the following statement to be true: The edit distance between the strings $x$ and $y$ is,
\begin{equation*}
    \delta(x,y) = \min_{x_1, x_2, x_3}(\delta(x_1,y_1)+\delta(x_2,y_2)+\delta(x_3,y_3)),
\end{equation*}
where $x_1, x_2, x_3$ ranges over all ordered partitions of $x$ and $x=x_1\bigcirc x_2\bigcirc x_3$.
\end{fact}

\begin{corollary}[of Fact~\ref{thm:AppFactdiffAlignment}]
\label{thm:AppCorJustX2}
The edit distance between the strings $x$ and $y$ is
\begin{equation*}
    \delta(x,y) = 2|x|+\min_{x_2}\delta(x_2,y_2),
\end{equation*}
such that $x_2$ ranges over all ordered partitions of $x$.
\end{corollary}
\begin{proof}
The strings $y_1$ and $y_3$ are strings of length $|x|$ (at least as large as $|x_1|$ and $|x_3|$) and consist of symbols that are not used in the entire string $x$. Therefore, the $\delta(x_1,y_1)=\delta(x_3,y_3)=|x|$ for any choice of $x_1$ and $x_3$.
\end{proof}

\begin{lemma} 
\label{thm:Appx2Format}
Given the two sequences $x$ and $y$ as mentioned above, there exists a substring $x_2$ of the form $5^T \ndots 6^T$ such that the $\delta(x,y)= 2|x|+\delta(x_2,y_2)$ when $y_2=\bigcirc_{b \in B}5^T\overline{G}(b)6^T$.
\end{lemma}
\begin{proof}
A proof by contradiction. Note that $x_2$ can be any substring of $x$. Let us assume that the (minimum) edit distance is achieved with $x_2$ that is \textbf{not} of the form $5^T\ndots6^T$, we then argue that one could change the format of $x_2$ to $5^T \ndots 6^T$ without increasing the cost. 

As we assume that $x_2$ is not of the form $5^T \ndots 6^T$ therefore it could be of any of the following forms: 
\begin{enumerate}
    \item Lets consider a scenario where $x_2$ is of the form $5^{\theta_w} \dots 6^T$ where $0 < \theta_w < T$, note that $y_2$ is already of the form $5^T \ndots 6^T$.
        \begin{equation*}
        \label{eq:AlignmentGadgetX21}
            x_2:= \overbrace{5^{\theta_w} a_i 6^T}^{w_1}\overbrace{5^T a_{i+1} 6^T \ndots 6^T}^{w_2}
        \end{equation*} 
        \begin{equation*}
        \label{eq:AlignmentGadgetY21}
            y_2:=\underbrace{5^T b_1 6^T5^T b_2 6^T \ndots 6^T}_{v_1 \bigcirc v_2}
        \end{equation*} 
        Recall from Fact~\ref{thm:AppFactdiffAlignment} that by fixing $w_1=5^{\theta_w}a_i6^T$ and $w_2=5^T a_{i+1} 6^T \ndots 6^T$ we have $\delta(x_2,y_2)=\min_{v_1, v_2} \delta(w_1,v_1)+\delta(w_2,v_2)$. 
        \begin{enumerate}
            \item Lets assume that the minimum is achieved when $v_1=5^Tb_16^{\theta_v}$, where $0 < \theta_v \leq T$. In such a scenario, $\delta(w_1,v_1)=\delta(5^{\theta_w}a_i6^T,5^Tb_16^{\theta_v})=\delta(a_i6^{T-\theta_v},5^{T-\theta_w}b_1)$. On the other hand we have $\delta(5^Ta_i6^T, v_1)=\delta(a_i6^{T-\theta_v},b_1)\leq \delta(a_i6^{T-\theta_v},5^{T-\theta_w}b_1)=\delta(w_1,v_1)$.\footnote{Consider three strings $s_1\in \Sigma^*$, and $s_2, s_3 \in \Gamma^*$, where $\Sigma$ and $\Gamma$ are two disjoint alphabet sets, i.e.~$\Sigma \cap \Gamma=\emptyset$, then $\delta(s_1\bigcirc s_2,s_3)\geq \delta(s_2,s_3)$ (similarly, $\delta(s_2,s_1\bigcirc s_3)\geq \delta(s_2,s_3)$). As the symbols in the string $s_1$ are different from symbols in $s_2$ and $s_3$ and the operations on symbols from $s_1$ will only be $delete$ or $substitute$. Therefore, one can get rid of the string $s_1$ by removing the symbols that got deleted (hence, reducing the cost) and by inserting the symbols in $s_3$ that otherwise would have been substituted by the symbols from $s_1$ (hence, maintaining the cost).} Therefore suggesting that if $v_1=5^Tb_16^{\theta_v}$ then setting $w_1=5^Ta_i6^T$ doesn't increase the cost.
            
            \item Lets assume that the minimum is achieved when $v_1=5^Tb_1^{\theta_v}$, where $0 < \theta_v \leq S_G$. $\delta(w_1,v_1)=\delta(5^{\theta_w}a_i6^T,5^Tb_1^{\theta_v})=\delta(a_i6^{T},5^{T-\theta_w}b_1^{\theta_v})\geq \delta(a_i6^{T},b_1^{\theta_v})=\delta(5^Ta_i6^T, v_1)$. Again suggesting that if $v_1=5^Tb_1^{\theta_v}$ then setting $w_1=5^Ta_i6^T$ cannot increase the cost.
            
            \item Lets assume that the minimum is achieved when $v_1=5^{\theta_v}$, where $0 \leq \theta_v \leq T$. $\delta(w_1,v_1)=\delta(5^{\theta_w}a_i6^T,5^{\theta_v})=\max(T+S_G,T+S_G+\theta_w-\theta_v)>\delta(\emptyset,5^{\theta_v})$. There by suggesting that if $v_1=5^{\theta_v}$ then setting $w_1=\emptyset$ will definitely cost less.
        \end{enumerate}
        This proves that no matter what form $v_1$\footnote{Note that we have not listed the scenario where $v_1$ is of the form $5^Tb_16^T*$. The reason being the following: If $v_1$ is of the form $5^Tb_16^T*$, then $v_2$ will be of the same form as of $x_2$ that we are arguing against.} is of, the minimum cost is achieved when $w_1=5^Ta_i6^T$ or when $w_1=\emptyset$, therefore, supporting the claim of Lemma \ref{thm:Appx2Format}. We use the same argument symmetrically to prove that $x_2$ cannot be of the form $5^T \ndots 6^{\theta_w}$, where $0< \theta_w < T$ or of the form $5^{\theta_{w_1}} \dots 6^{\theta_{w_2}}$ for $0< \theta_{w_1}, \theta_{w_2} < T$. 
    \item Consider the case where $x_2$ is of the form $a_i^{\theta_w} \dots 6^T$ where $0 < \theta_w \leq S_G$ and $y_2$ is of the form $5^T \ndots 6^T$.
        \begin{equation*}
        \label{eq:AlignmentGadgetX22}
            x_2:= \overbrace{a_i^{\theta_w} 6^T}^{w_1}\overbrace{5^T a_{i+1} 6^T \ndots 6^T}^{w_2}
        \end{equation*} 
        \begin{equation*}
        \label{eq:AlignmentGadgetY22}
            y_2:=\underbrace{5^T b_1 6^T5^T b_2 6^T \ndots 6^T}_{v_1 \bigcirc v_2}
        \end{equation*}
        \begin{enumerate}
            \item Lets assume that the minimum is achieved when $v_1=5^Tb_16^{\theta_v}$, where $0 < \theta_v \leq T$.
            Then, $\delta(w_1,v_1)=\delta(a_i^{\theta_w}6^T,5^Tb_16^{\theta_v})=\delta(a_i^{\theta_w}6^{T-\theta_v},5^Tb_1)$. The last $5$ of the substring $5^{T}b_1$ could either be substituted for a $6$ in $a_i^{\theta_w}6^{T-\theta_v}$ or could be substituted for a $\{0,1,2\}$ in $a_i^{\theta_w}$. Either way the $\delta(w_1,v_1) \geq S_G+(T-\theta_v) > Q+(T-\theta_v)\geq \delta(5^Ta_i6^T, v_1)$. Hence, suggesting that if $v_1=5^Tb_16^{\theta_v}$ then set $w_1=5^Ta_i6^T$.
            
            \item Lets assume that the minimum is achieved when $v_1=5^Tb_1^{\theta_v}$, where $0 < \theta_v \leq S_G$. $\delta(w_1,v_1)=\delta(a_i^{\theta_w} 6^T,5^Tb_1^{\theta_v})=T+\max(\theta_w,\theta_v)\geq T+\theta_v=\delta(\emptyset,v_1)$. Hence, suggesting that if $v_1=5^Tb_1^{\theta_v}$ then set $w_1=\emptyset$.
            
            \item Lets assume that the minimum is achieved when $v_1=5^{\theta_v}$, where $0 \leq \theta_v \leq T$. $\delta(w_1,v_1)=\delta(a_i^{\theta_w} 6^T,5^{\theta_v})>\delta(\emptyset,5^{\theta_v})$. Again suggesting that if $v_1=5^{\theta_v}$ then set $w_1=\emptyset$.
        \end{enumerate}   
        This again proves that no matter what form $v_1$ is of, the minimum cost is achieved when $w_1=5^Ta_i6^T$ or when $w_1=\emptyset$, therefore, supporting the claim of Lemma \ref{thm:Appx2Format}.  We use the same argument symmetrically to prove that $x_2$ cannot be of the form $5^T \ndots a_i^{\theta_w}$, where $0< \theta_w \leq S_G$ or of the form $a_i^{\theta_{w_1}} \dots a_j^{\theta_{w_2}}$ for $0< \theta_{w_1}, \theta_{w_2} \leq S_G$. 
    \item Lets consider a scenario where $x_2$ is of the form $6^{\theta_w}  \dots 6^T$ where $0 < \theta_w \leq T$.
        \begin{equation*}
        \label{eq:AlignmentGadgetX23}
            x_2:= \overbrace{6^{\theta_w}}^{w_1}\overbrace{5^T a_{i+1} 6^T  \ndots 6^T}^{w_2}
        \end{equation*} 
        \begin{equation*}
        \label{eq:AlignmentGadgetY23}
            y_2:=\underbrace{5^T b_1 6^T5^T b_2 6^T \ndots 6^T}_{v_1 \bigcirc v_2}
        \end{equation*}
        \begin{enumerate}
            \item Lets assume that the minimum is achieved when $v_1=5^Tb_16^{\theta_v}$, where $0 < \theta_v \leq T$. $\delta(w_1,v_1)=\delta(6^{\theta_w},5^Tb_16^{\theta_v})= \max(T+S_G,T+S_G+\theta_v-\theta_w) > Q+|\theta_w-\theta_v|\geq \delta(5^Ta_i6^{\theta_w},5^Tb_16^{\theta_v})$. Hence proving that if $v_1=5^Tb_16^{\theta_v}$ then set $w_1=5^Ta_i6^T$.
            
            \item Lets assume that the minimum is achieved when $v_1=5^Tb_1^{\theta_v}$, where $0 < \theta_v \leq S_G$. $\delta(w_1,v_1)=\delta( 6^{\theta_w},5^Tb_1^{\theta_v})=T+\theta_v=\delta(\emptyset,v_1)$. Hence proving that if $v_1=5^Tb_1^{\theta_v}$ then set $w_1=\emptyset$.
            
            \item Lets assume that the minimum is achieved when $v_1=5^{\theta_v}$, where $0 \leq \theta_v \leq T$. $\delta(w_1,v_1)=\delta(6^{\theta_w},5^{\theta_v})\geq \delta(\emptyset,5^{\theta_v})$. Hence again proving that if $v_1=5^{\theta_v}$ then set $w_1=\emptyset$.
        \end{enumerate}  
        This again proves that no matter what $v_1$ is, the minimum cost is achieved when $w_1=5^Ta_i6^T$ or when $w_1=\emptyset$, therefore, supporting the claim of Lemma \ref{thm:Appx2Format}. Again we use the same argument symmetrically to prove that $x_2$ cannot be of the form $5^T \ndots 5^{\theta_w}$, where $0< \theta_w \leq T$ or of the form $6^{\theta_{w_1}}  \dots 5^{\theta_{w_2}}$ for $0< \theta_{w_1}, \theta_{w_2} \leq T$. 
    \end{enumerate}
    Therefore, proving that there exists an $x_2$ of the form $5^T \ndots 6^T$ such that $\delta(x,y)=2|x|+\delta(x_2,y_2)$.
\end{proof} 
Using Corollary~\ref{thm:AppCorJustX2} and Lemma \ref{thm:Appx2Format} for a chosen partition of $y=y_1y_2y_3$ into three substrings,  we have shown that there exists a partition of $x=x_1x_2x_3$ such that $\delta(x,y)=2|x|+\delta(x_2,y_2)$ and $x_2$ is of the form $5^T  \ndots 6^T$. 

We still have to prove that there exists a coarse alignment $\coarseAlignC \in \mathcal{C}_{(3 \cdot 2^{n/2}-2), 2^{n/2}}$ such that $\delta(x_2,y_2)=\sum_{(i,j)\in \coarseAlignC}\delta(u_i,v_j)$ where $u_i$ and $v_j$ are substrings of $x_2$ and $y_2$ as mentioned above. To prove that we use the result from the Lemma \ref{thm:AppSeparatorAligns} below.

\begin{lemma}
\label{thm:AppSeparatorAligns}
Given two substrings $x_2$ and $y_2$, both of the form $5^T \ndots 6^T$. There exists a separator $6^T5^T$ in the string $x_2$ (assuming that $x_2$ has one) that completely aligns with a separator $6^T5^T$ from the other string $y_2$ (also assuming that $y_2$ has one). A separator is a substring $6^T5^T$ that repeatedly occurs in both the strings $x_2$ and $y_2$.
\end{lemma}
\begin{proof}
We have already established that $x_2$ is of the form $5^T  \ndots 6^T$, and earlier we chose $y_2$ to be of the form $5^T  \ndots 6^T$. Let,
\begin{equation*}
    x_2=5^Ta'_16^T5^Ta'_2\overbrace{6^T5^T}^{\text{sep }w}a'_36^T  \ndots 5^Ta'_{t_a}6^T,
\end{equation*}
\begin{equation*}
    y_2=5^Tb'_16^T5^Tb'_26^T5^Tb'_3\underbrace{6^T5^T}_{\text{sep}}b'_46^T  \ndots 5^Tb'_{t_b}6^T,
\end{equation*}
such that $\forall i \in [t_a], \forall j \in [t_b], a'_i \in \{0,1,2\}^*$ and $b'_j \in \{0,1,2\}^*$ and $ Q < |a'_i|=|b'_j|=S_G << T$. Without loss of generality, lets assume $t_a \leq t_b$. Then, $\delta(x_2,y_2)\geq 2T(t_b-t_a)+S_G(t_b-t_a)$, because there will definitely be $2T(t_b-t_a)+S_G(t_b-t_a)$ number of symbols inserted to convert the string $x_2$ to $y_2$. Also, it is easy to see that $\delta(x_2,y_2)\leq 2T(t_b-t_a)+S_G(t_b-t_a)+Q\cdot t_a$
We will now show that there is no such optimal alignment of symbols where there is no separator in $x_2$ that completely aligns with a separator in $y_2$.
\newline
A proof by contradiction. Let us assume that no separator in $x_2$ aligns with any separator in $y_2$. Let $w=6^T5^T$ be a separator from $x_2$ and let $w$ align with a substring $v$ from $y_2$ in an optimal alignment. The substring $v$ can be of the following forms:
\begin{enumerate}
    \item Let $v=b_i^{\theta_1}6^T5^Tb_{i+1}^{\theta_2}$ with $0 \leq \theta_1, \theta_2 \leq S_G$. This is a trivial case leading to contradiction.
    \item Let $v=b_i^{\theta_1}6^{\theta_2}$, with $0 \leq \theta_1 \leq S_G, 0\leq \theta_2 \leq T$. The $\delta(w,v)=\delta(6^T5^T,b_i^{\theta_1}6^{\theta_2})> 2T-\theta_1-\theta_2$. Note that this is a deletion cost because of the mismatch in the number of symbols in $w$ and $v$. As $\theta_1 \leq S_G$ and $\theta_2 \leq T$. The deletion cost $\geq T-S_G$. When $v=5^{\theta_1}b_i^{\theta_2}$ with $0\leq \theta_1 \leq T, 0 \leq \theta_2 \leq S_G$ we follow the same argument and get deletion cost $\geq T-S_G$.
    \item Let $v=6^{\theta_1}5^{\theta_2}$, with $0 \leq \theta_1, \theta_2 < T$. As the separator $w$ aligns with the substring $v$, that means the $6^{T- \theta_1}$s that are \emph{prefixing} $v$ has to be either inserted or substituted by some symbols $\in \{0,1,2\}$, or matched with the $6$s of a separator appearing before $w$ in the string $x_2$. For the 6s that get inserted, it is cheaper to match the 6s with the $6$s that get deleted instead in $\delta(6^T5^T,6^{\theta_1}5^{\theta_2})$ =$\delta(6^{T-\theta_1}5^{T-\theta_2},\emptyset)$. Same argument holds for the $5$s that are \emph{suffixing} $v$. Suppose these 6s are substituted from the 5s or the $\{0,1,2\}$ again it costs the same to just delete these and match freely with the 6s that are getting deleted in the alignment of $w$ and $v$. And say for any reason some 6 in the prefix gets matched with a 6 from a separator preceding $w$, then the deletion plus the substitution cost is $> T+S_G$. Therefore, it is cheaper to align the separator $w$ with the separator from $y_2$ that is surrounding $v$.
    \item Let $v=6^{\theta_1}5^Tb_i^{\theta_2}$, with $0 \leq \theta_1 < T, 0\leq \theta_2 \leq S_G$. Similar to the argument in item 3, we analyze the cost to generate the substring $6^{T- \theta_1}$s that is prefixing $v$. Even in this case the deletion plus substitution cost of not aligning the separators is $> T+S_G$. A similar argument holds when $v=b_i^{\theta_1}6^T5^{\theta_2}$, with $0 \leq \theta_1 \leq S_G, 0\leq \theta_2 < T$.
    \item Let $v=5^{\theta_1}b_i6^{\theta_2}$, with $0 \leq \theta_1,\theta_2 \leq T$. Consider the last 6 of the string $w$. Whether this 6 matches with a 6 from $6^{\theta_2}$ in $v$ or substitutes any symbol from the $b_i$ or $5^\theta_1$ in $v$ the deletion and the substitution cost is $> T$. 
    \item Let $v=5^{\theta_1}b_i6^T5^{\theta_2}$, with $0 \leq \theta_1 \leq T$, $0 \leq \theta_2 < T$. The argument here is similar to that of the argument in item 3, where we analyze the cost of generating the substring $5^{T-\theta_2}$ which succeeds the substring $v$ in $y_2$. Either we align the separators completely or pay a deletion plus substitution cost $>T+S_G$. Same argument can be used when $v=6^{\theta_1}5^Tb_i6^{\theta_2}$, with $0 \leq \theta_1 < T$, $0 \leq \theta_2 \leq T$. Just that, here we analyze the cost of generating the substring $6^{T-\theta_1}$ which precedes $v$.
    \item Let $v=6^{\theta_1}5^Tb_i6^T5^{\theta_2}$, with $0 \leq \theta_1,\theta_2 < T$. The deletion and substitution cost to generate the substring $6^{T-\theta_1}$ which is a prefix to the substring $v$ (or $5^{T-\theta_2}$ a suffix to $v$) is $>T+S_G$. 
    
\end{enumerate}
The deletion and substitution cost induced when a separator $6^T5^T$ from $x_2$ doesn't align with a separator from $y_2$ is $\geq S_G$\footnote{The deletion and substitution costs induced when a separator $6^T5^T$ from $x_2$ doesn't align with a separator from $y_2$ is either $\geq T-S_G$ or $> T$ or $> T +S_G$. As $S_G << T$, the deletion and substitution costs are always higher than $S_G$.}. As there are a total of $(t_a -1)$ such separators in $x_2$, if no separator from $x_2$ aligns with any separator from $y_2$ the total cost of the transformation then becomes $\geq 2T(t_b-t_a)+S_G(t_b-t_a) + S_G \cdot (t_a-1)$. Therefore, such a transformation is not an optimal transformation because the edit distance $\delta(x_2,y_2)\leq 2T(t_b-t_a)+S_G(t_b-t_a)+Q\cdot t_a$ and as long as $t_a > 1$ we have $S_G\cdot (t_a-1)>Q\cdot t_a$\footnote{Recall that, $S_G=|G(\cdot)|=|\overline{G}(\cdot)|$ and $Q=\delta(G(a),\overline{G}(b))$ such that $(a,b)$ wasn't a satisfying assignment. Clearly, $S_G>Q$. Also one can chose to make the gadgets in such a way that $S_G(t-1)>Q\cdot t$ for all $t>1$.} Thus, we prove that in an optimal transformation there will always exist a separator in $x_2$ that will freely align with a separator in $y_2$.
\end{proof}
As proved in Lemma \ref{thm:AppSeparatorAligns}, the existence of a completely aligning separator pair in an optimal alignment lets us to make the following statement: $\delta(x_2,y_2)=\delta(x',y')+\delta(x'',y'')$, such that $x',x'',y',y''$ are all of the form $5^T \ndots 6^T$. 
\begin{equation*}
    x_2=\overbrace{5^Ta'_16^T5^Ta'_2 \ndots 6^T}^{x'}\overbrace{5^Ta'_i6^T \ndots 5^Ta'_{t_a}6^T}^{x''},
\end{equation*}
\begin{equation*}
    y_2=\underbrace{5^Tb'_16^T5^Tb'_26^T5^Tb'_3 \ndots 6^T}_{y'}\underbrace{5^Tb'_j6^T  \ndots 5^Tb'_{t_b}6^T}_{y''},
\end{equation*}
Using this argument recursively we can see that the strings $x_2$ and $y_2$ gets partitioned into substrings of the form $5^T \ndots 6^T$ such that $\delta(x_2,y_2)$ is the sum of pair wise edit distance of these substrings, where only one of the substrings in each pair contains more than zero separators. This proves the claim of our Lemma \ref{thm:AppEditDistanceCourseAlignment}.
\end{proof} 

Using the result of the above mentioned Lemma \ref{thm:AppEditDistanceCourseAlignment} we have provided an insight on the relation between edit distance of sequences $x$ and $y$ and edit distance between the \emph{gadgets} of $x$ and \emph{gadgets} of $y$. We now prove the correctness of our reduction. 

\begin{lemma}
\label{thm:AppCorrectnessOfReduction}
For every $n$, there exists a constant $C^* \in \mathbb{Z}$ such that
\begin{equation*}
    \delta(x,y) < C^*
\end{equation*}
if and only if $\propertyPPedit(\matrixM^{\truthtable(S)})=1$.
\end{lemma}     
\begin{proof}
If $\propertyPPedit(\matrixM^{\truthtable(S)})=1$ then it implies $\propPedit(\matrixM^{\truthtable(S)},Q)=1$ because of the promise that the branching program $S$ belongs to the set $\mathcal{S}$. The statement $\propPedit(\matrixM^{\truthtable(S)},Q)=1$ implies that there exists a $path$ $\pathP \in \setOfPaths_{2L-1,L}$ (here $L=2^{n/2}$) such that even with the highest allowed jump parameter $\mu=Q$ the path-cost is $\pathCost(\matrixM^{\truthtable(S)},\pathP, Q) < T_r$, for a fixed threshold $T_r$. Which invariably means that $\forall \mu$, $\pathCost(\matrixM^{\truthtable(S)},\pathP, \mu) < T_r$. Using the Algorithm \ref{Alg:PtoCmain} that generates a  $\text{coarse alignment}$ $\coarseAlignC \in \mathcal{C}_{3L-2,L}$ for a given path $\pathP \in \setOfPaths_{2L-1,L}$ as an input, we get a coarse alignment $\coarseAlignC$ corresponding to path $\pathP$. We get $\editCost(\coarseAlignC)< T_r$ because $\forall \mu, \pathCost(\matrixM^{\truthtable(S)},\pathP, \mu) < T_r$.\footnote{This happens due to the choices of $C_0$, $C_1$, $C_{jump}$ that were made in Definition \ref{def:PP_EditProperty}.} Also, from Lemma \ref{thm:AppEditDistanceCourseAlignment} in Appendix \ref{sec:AppReduction} we have that $\delta(x,y)=2|x| + \min_{\coarseAlignC \in \mathcal{C}_{3 L-2, L}} \editCost(\coarseAlignC)$. This implies $\delta(x,y)< 2|x| + T_r$ and, we set our constant $C^*=2|x|+T_r$. 

We now prove the other direction. If $\propertyPPedit(\matrixM^{\truthtable(S)})=0$ it implies that $\propPedit(\matrixM^{\truthtable(S)},0)=0$ as the branching program $S\in \mathcal{S}$. Which in turn implies that $\forall \mu, \forall \pathP \in \setOfPaths_{2L-1,L}$, $\pathCost(\matrixM^{\truthtable(S)},\pathP,\mu)\geq T_r$. Using the result from Lemma~\ref{thm:AppPtoCmain} and \ref{thm:AppEitherPathorD} in Appendix~\ref{sec:AppReduction} we show that if $\forall \mu, \forall \pathP \in \setOfPaths_{2L-1,L}$, $\pathCost(\matrixM^{\truthtable(S)},\pathP,\mu)\geq T_r$ then $\forall \coarseAlignC \in \mathcal{C}_{3L-2, L}$, $\editCost(\coarseAlignC) \geq T_r$ which implies $\delta(x,y)\geq 2|x| + T_r$. Thus, implying $\delta(x,y) \geq C^*$.
\end{proof}

\begin{algorithm}[h]
\label{Alg:PtoCmain}
\SetAlgoLined
\KwResult{Given an input a $\textit{path}$ $\pathP \in \setOfPaths_{2L-1,L}$, generate a $\textit{coarse alignment}$ $\coarseAlignC \in \mathcal{C}_{3L-2,L}$.}
 C=[], i=1,j=1, k=|P|\;
 \While{(i $\leq$ k)}{
  (a+1,b)=P[i]\;
  (c+1,d)=P[i+1]\;
  
  \eIf{(a $\neq$ c)}{
   \eIf{(d=b)}{
   C[j]=$((a+b, \dots ,c+d),d)$\;
   j=j+1\;
   }{
   C[j]=$(c+d,(b, \dots ,d))$\;
   j=j+1}
   i=i+2\;
   }{
   C[j]=$(a+b,b)$\;
   j=j+1\;
   i=i+1\;
  }
 }
 \Return C\;
\caption{Convert a given $path$ $\pathP$ to a $\text{coarse alignment}$ $\coarseAlignC$.}
\end{algorithm}

\begin{lemma}
\label{thm:AppPtoCmain}
Algorithm \ref{Alg:PtoCmain} when given a $path$ $\pathP=((i_1,j_1),(i_2,j_2), \dots, (i_k,j_k)) \in$ $\setOfPaths_{2L-1,L}$ as an input, outputs a coarse alignment $\coarseAlignC \in \mathcal{C}_{3L-2,L}$, such that $|\coarseAlignC| \leq |\pathP|$.
\end{lemma}
\begin{proof}
We provide a simple Algorithm \ref{Alg:PtoCmain} that when given as an input a $\textit{path}$ $\pathP \in \setOfPaths_{2L-1,L}$, it generates a sequence $\coarseAlignC=((p_1,q_1),(p_2,q_2), \dots ,(p_m,q_m))$. It is easy to see from the algorithm that $k=|\pathP|\geq m=|\coarseAlignC|$. We now show that the sequence $\coarseAlignC$ generated by this algorithm indeed is a $\textit{coarse alignment}$ $\coarseAlignC \in \mathcal{C}_{3L-2,L}$. 

For all neighbouring pairs $(i_l,j_l),(i_{l+1},j_{l+1}) \in \pathP$ that the algorithm reads as inputs it checks if there is a \textit{jump}\footnote{Refer to the Definition \ref{def:PP_EditProperty} for the definitions of a \textit{path} and also \textit{jumps} in a path.} between these neighbouring pairs. If there is no jump in the path $\pathP$ at $(i_l,j_l)$ and $(i_{l+1},j_{l+1})$, then the algorithm just adds a term $(p_*,q_*)$ to the sequence $\coarseAlignC$ such that $|p_*|=|q_*|=1$ and $p_*=i_l+j_l-1$ and $q_*=j_l$ and changes the position of the pointer to the next term. But, if there is a jump in the path $\pathP$ at $(i_l,j_l)$ and $(i_{l+1},j_{l+1})$ then the algorithm checks whether the jump is to a row above or to a row below. When the jump is to a row below then the algorithm adds a sequence $(p_*,q_*)$ such that $p_*=(i_l+j_l-1, \dots,i_{l+1}+j_{l+1}-1)$ ensuring that $|p_*|$ is the number of rows jumped below and $q_*=j_l=j_{l+1}$ ensuring that $|q_*|=1$ and changes the pointer to the next but one term. When the jump is to a row above then the algorithm adds a sequence $(p_*,q_*)$ such that $q_*=(j_l, \dots,j_{l+1})$ ensuring that $|q_*|$ is the number of rows jumped above and $p_*=i_l+j_l+1=i_{l+1}+j_{l+1}+1$\footnote{According to the definition of a $path$, when a jump is to a row above then $i_l+j_l=i_{l+1}+j_{l+1}$.} ensuring that $|p_*|=1$ and then changes the pointer to the next but one term. 

As the definition of a $path$ requires that $1=j_1 < j_2 < \dots<j_k=L$. Therefore, it is easy to see that $\forall r \in [m-1], q_{r} <_e q_{r+1}$\footnote{Given two sequences $a$ and $b$, $a <_e b$ implies $\forall u \in a, \forall v \in b, u<v$.} and $\forall r,s \in [m]$, $q_r \cap q_s=\emptyset$ and $\cup_{r=1}^{m}q_r=[L]$. Also using the same definition we know that either $i_l=i_{l+1}$ (implying $j_{l+1}=j_l +1$) or $i_l<i_{l+1}$ (implying $j_{l+1}=j_l$) or $i_l>i_{l+1}$ (implying $j_{l+1}=j_l +i_l-i_{l+1}$) therefore, ensuring that $i_{l+1}+j_{l+1}\geq i_l+j_l$ therefore proving that $\forall r \in [m-1], p_{r} <_e p_{r+1}$ and $\forall r,s \in [m]$, $p_r \cap p_s=\emptyset$ and $\cup_{r=1}^{m}p_r=[(i_1+j_1-1) \dots (i_{k}+j_{k}-1)]$.

It is given that $\pathP\in \setOfPaths_{2L-1,L}$ that implies $1 \leq i_1 \leq 2L-1, j_1=1$ and $1 \leq i_k \leq 2L-1, j_k=L$. Therefore, it is now clear that the sequence $\coarseAlignC$ produced by the Algorithm \ref{Alg:PtoCmain} is indeed a coarse alignment $\coarseAlignC \in \mathcal{C}_{3L-2,L}$.
\end{proof}

\begin{lemma}
\label{thm:AppEitherPathorD}
Algorithm \ref{Alg:CtoPmain} when given a coarse alignment $\coarseAlignC \in \mathcal{C}_{3L-2,L}$ (with $L=2^{n/2}$) as an input outputs a sequence $\pathP$. This sequence $\pathP$ is \textbf{either} a path $\pathP \in \setOfPaths_{2L-1,L}$ \textbf{or} there exists another coarse alignment $\coarseAlignD \in \mathcal{C}_{3L-2,L}$ for which a path $\pathR \in \setOfPaths_{2L-1,L}$ can be generated using Algorithm \ref{Alg:CtoPmain} and the $\editCost(\coarseAlignD) \leq \editCost(\coarseAlignC)$.
\end{lemma}
\begin{proof}
Let $\coarseAlignC=((p_1,q_1),(p_2,q_2), \dots ,(p_m,q_m))$ be a coarse alignment from the set $\mathcal{C}_{3L-2,L}$. We classify every element $(i,j) \in \coarseAlignC$ into $bad$ and $good$ terms in the following way:
\begin{equation*}
    \bad(i,j)=\begin{cases}
    1, & \text{if } \exists a \in i \text{ and } \exists b \in j \text{ such that } (a-b)<0 \text{ or } (a-b)\geq 2L-1\\
    0 & \text{otherwise}.
  \end{cases}
\end{equation*}
The $\editCost$ for alignment $\coarseAlignC$ will be:

\begin{equation*} \label{eq1}
\begin{split}
\editCost(\coarseAlignC) & = \Sigma_{(i,j) \in C} \delta(u_i,v_j) \\
 & = \underbrace{\Sigma_{(i,j) \in \coarseAlignC, (\bad(i,j)=1)} \delta(u_i,v_j)}_{\text{bad terms}} + \underbrace{\Sigma_{(i,j) \in \coarseAlignC, (\bad(i,j)=0)} \delta(u_i,v_j)}_{\text{good terms}},
\end{split}
\end{equation*}
where $u_i=\bigcirc_{l' \in i}5^Tg_{l'}6^T$\footnote{Note that $g_{l'}=G(a_{l'+1-2^{n/2}})$ when $0 \leq (l'-2^{n/2}) < 2^{n/2}$ and $g_{l'}=r$ otherwise.} and $v_j=\bigcirc_{l'' \in j}5^T\overline{G}(b_{l''})6^T$.

\begin{fact}
\label{thm:AppFactDlessC}
For every coarse alignment $\coarseAlignC \in \mathcal{C}_{3L-2,L}$ that contains bad terms, there exists a coarse alignment $\coarseAlignD \in \mathcal{C}_{3L-2,L}$ such that $\coarseAlignD$ consists of good terms and the $\editCost(\coarseAlignD) \leq \editCost(\coarseAlignC)$.
\end{fact}
\begin{proof}
A bad term $(i,j)$ in the coarse alignment $\coarseAlignC$ implies $\exists a \in i, \exists b \in j$ such that $(a-b) <0$ or $(a-b) \geq 2L-1$. Let $K=2L-1$. This term can be a bad term in three different ways: 
\begin{enumerate}
    \item Category 1: $\forall a \in i, \forall b \in j$, either $(a-b)<0$ or $0 \leq (a-b) < K$.
    \item Category 2: $\forall a \in i, \forall b \in j$, either $(a-b)\geq K$ or $0 \leq (a-b) < K$.
    \item Category 3: $\exists a \in i, \exists c \in i$, such that $(a-j) < 0$ and $(c-j)\geq K$.
\end{enumerate}
Consider the coarse alignment $\coarseAlignN=((1,1),(2,2), \dots,(L,L))$. 
Clearly the $\editCost(\coarseAlignN)\leq L\cdot Q$ and the corresponding path for $\coarseAlignN$ using the Algorithm \ref{Alg:CtoPmain} is $\pathP_N=((1,1),(1,2), \dots,(1,L))$. Let $\coarseAlignC=((p_1,q_1),(p_2,q_2), \dots,(p_m,q_m))$ and let us label each of these terms into bad or good terms. Suppose $\coarseAlignC$ contains a bad term $(p,q)$ of category 3, then $|p|\neq 1$ because bad category 3 requires that the following condition is met: $\exists a \in p$ such that $(a-q)<0$ and $\exists b \in p$ such that $(b-q)\geq K$. Combining both these conditions we get $(b-a) \geq K$. The $\editCost(\coarseAlignC)\geq \delta(u_p,v_q) > K\cdot (2T+S_G) > \editCost(\coarseAlignN)$ which proves this fact. Therefore, we can safely only consider cases where the coarse alignment consists of bad terms of category 1 and 2.

Let $\groupG_1=((p_{k'},q_{k'}),(p_{k'+1},q_{k'+1}),\dots,(p_{k''},q_{k''}))$ be the first group of bad terms in $\coarseAlignC$. It is easy to see that the bad terms come in groups of 2 or more. This means that the entire preceding group $\groupG_0=((p_1,q_1),(p_2,q_2),\dots,(p_{k'-1},q_{k'-1}))$ and the next term $\groupG_2=((p_{k''+1},q_{k''+1}))$ has to be good. We now claim that there exists a group $\groupG'=((p'_1,q'_1),(p'_2,q'_2),\dots,(p'_{k'''},q'_{k'''}))$ that covers the set of indices of $\groupG_1$ and has all good terms and the $\editCost(\groupG') \leq \editCost(\groupG_1)$. As the terms $(p_{k'},q_{k'})$ and $(p_{k''},q_{k''})$ of $\groupG_1$ can be of any of the two bad categories we have a total of four cases to consider:
\begin{enumerate}
    \item Let $(p_{k'},q_{k'})$ of category 1 and $(p_{k''},q_{k''})$ of category 2: This scenario suggests that there will be two neighbouring terms $(p_{l},q_{l})$ and $(p_{l+1},q_{l+1})$ such that they are of bad category 1 and 2 respectively. As $(p_{l},q_{l})$ is of category 1, that implies $\exists a \in p_{l}$ and $\exists b \in q_{l}$ such that $(a-b) < 0$. Similarly, as $(p_{l+1},q_{l+1})$ is of category 2, that implies $\exists c \in p_{l+1}$ and $\exists d \in q_{l+1}$ such that $(c-d) \geq K = 2L-1$. Combining these two inequalities we get $(c-a)-(d-b)>K$. As $(p_{l},q_{l})$ and $(p_{l+1},q_{l+1})$ are neighbouring terms in a coarse alignment the $\editCost(((p_{l},q_{l}),(p_{l+1},q_{l+1})))> \delta(u_{p_l}\bigcirc u_{p_{l+1}},v_{q_l}\bigcirc v_{q_{l+1}}) > (2T + S_G)\cdot K > \editCost(\coarseAlignN)$. Therefore proving this fact.
    
    \item Let $(p_{k'},q_{k'})$ of category 2 and $(p_{k''},q_{k''})$ of category 1: This scenario doesn't exist because of the following reason. Let $(p_{l},q_{l})$ and $(p_{l+1},q_{l+1})$ be two neighbouring terms such that they are of bad category 2 and 1 respectively. This implies $\exists a \in p_{l}$ and $\exists b \in q_{l}$ such that $(a-b) \geq K= 2L-1$. Similarly, as $(p_{l+1},q_{l+1})$ is of category 1, that implies $\exists c \in p_{l+1}$ and $\exists d \in q_{l+1}$ such that $(c-d)<0$. Combining both these inequalities we get $(a-c)-(b-d)>K$. As $(p_{l},q_{l})$ and $(p_{l+1},q_{l+1})$ are elements of a coarse alignment both $(a-c)$ and $(b-d)$ will be negative. That implies $|b-d|>K=2L-1$ which is not possible because the indices in $q_{*}$ ranges between $1 \dots L$.
    \item Let both $(p_{k'},q_{k'})$ and $(p_{k''},q_{k''})$ be of category 1: We first claim that in this scenario all the intermediate bad terms in the group $\groupG_1$ will also be of category 1 because of the impossibility result from scenario 2. 
    
    Let $(p_l,q_l),(p_m,q_m)\in \groupG_1$ be two nearest terms of the form $|p_l|= 1, |q_l|\neq1$ and $|p_m|\neq1, |q_m|=1$. The only other intermediate terms in between these terms are of the form $(p_{*},q_{*})$ such that $|p_{*}|=|q_{*}|=1$ where $\editCost((p_{*},q_{*}))=Q$\footnote{The $\editCost((p_{*},q_{*}))=\delta(u_{p_*},v_{q_*})=\delta(5^Tg_{p_*}6^T,5^T\overline{G}(b_{q_*})6^T)=\delta(5^Tr6^T,5^T\overline{G}(b_{q_*})6^T)$ because as $p_* < L$ we have $g_{p_*}=r$}. We now do the following: w.l.o.g. lets assume $|q_l|\geq |p_m|$. We remove the $(p_m-1)$ maximum most elements from the set $q_l$ and $(p_m-1)$ minimum most elements from the set $p_m$. For an element that we remove from the set $q_l$ we pair it with an element that we have removed from the set $p_m$. Thereby reducing the total cost by a positive quantity\footnote{The minimum gain here is $2\cdot(2T+S_G)-3Q$ which is positive because $T>>S_G>Q$.}. We know that such pairs exists in $\groupG_1$ because the group $\groupG_0$ and $\groupG_2$ only consists of good terms. We keep repeating this process until we get rid of all the pairs like $(p_l,q_l)$ and $(p_m,q_m)$. Also note that this process reduces the edit-cost. Therefore, by following the procedure repeatedly we have converted the group of bad terms $\groupG_1$ into a new group $\groupG'$ which spans all the indices spanned by $\groupG_1$ and also has the $\leq \editCost(\groupG_1)$. And because we have got rid of all the pairs $(p_l,q_l),(p_m,q_m)$ therefore, either all the terms in $\groupG'$ are of the form $|p_*|= 1, |q_*|\neq1$ and  $|p_*|=1, |q_*| = 1$ or are of the form $|p_*|\neq 1, |q_*|=1$ and  $|p_*|=1, |q_*| = 1$.
    
    We now have to prove that all the terms in $\groupG'$ are good terms. Let $\Delta'=\max(p_{k'-1})-\max(q_{k'-1})$ and $\Delta''=\min(p_{k''+1})-\min(q_{k''+1})$. These $0 \leq \Delta' < K$ and $0 \leq \Delta'' < K$ as the $(p_{k'-1},q_{k'-1})$ and $(p_{k''+1},q_{k''+1})$ were the good terms outside of $G_1$. The difference between the number of indices spanned by the $p_*$ terms and the number of indices spanned by the $q_*$ term in $\groupG_1$ is $\Delta''-\Delta'$. Two cases to consider again:
    \begin{enumerate}
        \item Case $\Delta' > \Delta''$: The group $\groupG'$ should consist of a term where $(p',q')$ such that $|p'|=1$ and $|q'|\neq 1$ because the number of indices spanned by the $q'_*$ terms is higher than the number of indices spanned by the $p'_*$ terms. Let $\groupG'=((p'_1,q'_1),(p'_2,q'_2),\dots, (p'_l,q'_l))$ and let $\forall i \in [l], \Delta_i=p'_i-\max(q'_i)$ therefore, $\Delta_1 \leq \Delta'$ and $\Delta_l=\Delta''$. As all the other terms will also be either of the form $|p'_*|=1$ and $|q'_*|\neq 1$ or of the form $|p'_*|= 1$ and $|q'_*|=1$ (as proved in the previous paragraph) therefore, $\forall i \in [l-1], \Delta_{i+1} \leq \Delta_i$. Therefore, we have $\Delta' \geq \Delta_1 \geq \Delta_2 \geq \dots \geq \Delta_l=\Delta''$ implying that $\forall i \in [l], \text{ bad}(p'_i,q'_i)=0$.
        
        \item Case $\Delta'' > \Delta'$: The group $\groupG'$ should consist of a term where $(p',q')$ such that $|p'| \neq 1$ and $|q'|=1$ because the number of indices spanned by the $p_*$ terms is higher than the number of indices spanned by the $q_*$ terms. All the other terms will also be either of the form $|p_*| \neq 1$ and $|q_*|=1$ or $|p_*|= 1$ and $|q_*|=1$. Let $\groupG'=((p'_1,q'_1),(p'_2,q'_2), \dots,(p'_l,q'_l))$ and let $\forall i \in [l], \Delta_i=\min(p'_i)-q'_i$ therefore, $\Delta_1=\Delta'$ and $\Delta_l\leq \Delta''$. As all the other terms will also be either of the form $|p'_*| \neq 1$ and $|q'_*|= 1$ or of the form $|p'_*|= 1$ and $|q'_*|=1$ (as proved in the previous paragraph) therefore, $\forall i \in [l-1], \Delta_{i+1} \geq \Delta_i$. Therefore, we have $\Delta'= \Delta_1 \leq \Delta_2 \leq \dots \leq \Delta_l \leq \Delta''$ implying that $\forall i \in [l], \text{ bad}(p'_i,q'_i)=0$.
        
        \item Case $\Delta'' = \Delta'$: All the terms in the group $\groupG'$ should be of the form $|p'| = 1$ and $|q'|=1$ because the number of indices spanned by the $p_*$ terms is same as the number of indices spanned by the $q_*$ terms. Let $\groupG'=((p'_1,q'_1),(p'_2,q'_2), \dots,(p'_l,q'_l))$ and let $\forall i \in [l], \Delta_i=p'_i-q'_i$ therefore, $\Delta_1=\Delta'$ and $\Delta_l= \Delta''$. Also, $\forall i \in [l-1], \Delta_{i+1} = \Delta_i$ because all the terms in the group $\groupG'$ are of the form $|p'| = 1$ and $|q'|=1$. Therefore, $\Delta'= \Delta_1 = \Delta_2 =\dots= \Delta_l = \Delta''$ implying that that $\forall i \in [l], \text{ bad}(p'_i,q'_i)=0$.
    \end{enumerate}
    \item Let both $(p_{k'},q_{k'})$ and $(p_{k''},q_{k''})$ be of category 2: In this scenario all the intermediate bad terms in the group $\groupG_1$ will also be of category 2 because of the impossibility result from scenario 2. And the rest of the argument is same as in the scenario 3.
\end{enumerate}

Note that all the above mentioned steps were to analyze the first group of bad terms. We keep repeating this procedure till we arrive at a coarse alignment $\coarseAlignD$ that only contains good terms. As we see that the procedure mentioned above never increases the $\editCost$ we can therefore safely say that $\editCost(\coarseAlignD) \leq \editCost(\coarseAlignC)$.
\end{proof}

\begin{algorithm}[h]
\label{Alg:CtoPmain}
\SetAlgoLined
\KwResult{Given a $\textit{coarse alignment}$ $\coarseAlignC \in \mathcal{C}_{}$ generate a sequence $\pathP$.}
 P=\{\}, i=0, j=0, m=|C|\;
 \While{(i<m)}{
  (p,q)=C[i]\;
  \eIf{(|p| $\neq$ 1 $\text{or}$ |q| $\neq$ 1)}{
   \eIf{(|p| $\neq$ 1)}{
   P[j]=$(\min(p)-q+1,q)$\;
   P[j+1]=$(\max(p)-q+1,q)$\;
   j=j+2\;
   }{
   P[j]=$(p-\min(q)+1,\min(q))$\;
   P[j+1]=$(p-\max(q)+1,\max(q))$\;
   j=j+2\;
   }
   }{
   P[j]=$(p-q+1,q)$\;
   j=j+1\;
  }
  i=i+1\;
 }
 \Return P\;
\caption{Convert a given $\text{coarse alignment}$ $\coarseAlignC$ to a sequence $\pathP$.}
\end{algorithm}

\begin{fact}
\label{thm:AppFactGoodtermImpliesPath}
The Algorithm~\ref{Alg:CtoPmain} outputs a path $\pathR \in \setOfPaths_{2L-1,L}$ when the input is a coarse alignment $\coarseAlignD \in \mathcal{C}_{3L-2,L}$ containing only good terms.
\end{fact}
\begin{proof}
Apply the Algorithm \ref{Alg:CtoPmain} on the coarse alignment $\coarseAlignD=((p_1,q_1),(p_2,q_2), \dots,(p_m,q_m))$ as input and let the output sequence be $\pathR=((i_1,j_1),(i_2,j_2), \dots,(i_k,j_k))$. We will now prove that $\pathR \in \setOfPaths_{2L-1,L}$ when $\coarseAlignD$ contains only good terms.

Given an input $\coarseAlignD=((p_1,q_1),(p_2,q_2),\dots,(p_m,q_m))$, the Algorithm \ref{Alg:CtoPmain} checks each term $(p_l,q_l), \forall l \in [m]$ and creates two (or one) new terms $(\min(p_l)-\min(q_l)+1, \min(q_l))$ and $(\max(p_l)-\max(q_l)+1, \max(q_l)), \forall l \in [m]$ and creates the sequence $\pathR$. As $\coarseAlignD$ contains all good terms, it is clear that $\forall r \in [k]$, we have $1 \leq i_r \leq K=2L-1$ and $1 \leq j_r \leq L$. Using the definition of a \textit{coarse alignment} (Definition \ref{def:course_alignment}) we know that $\forall l \in [m-1], q_l <_e q_{l+1}$\footnote{Given two sequences $a$ and $b$, $a <_e b$ implies $\forall u \in a, \forall v \in b, u<v$.}, and $\cup_{i=1}^m q_i=[L]$ therefore, we have $1=j_1 \leq j_2 \leq \dots \leq j_k=L$. 

Consider the term $(p_l,q_l) \in \coarseAlignD$, the algorithm generates the following terms $(\min(p_l)-\min(q_l)+1, \min(q_l))$, $(\max(p_l)-\max(q_l)+1, \max(q_l))$\footnote{Note that if $|p_l|=|q_l|=1$ then both the terms are same and the algorithm just adds one term.} for the sequence $\pathR$. If $|p_l| \neq 1$ then the algorithm generates $(\min(p_l)-q_l+1, q_l)$, $(\max(p_l)-q_l+1, q_l)$, clearly generating two terms $(i_r,j_r),(i_{r+1},j_{r+1}) \in R$ such that $i_{r+1} > i_r$ while $j_{r+1}=j_r$ thus, satisfying another condition of a \textit{path}. Also, when $|q_l| \neq 1$ then the algorithm generates $(p_l-\min(q_l)+1, \min(q_l))$, $(p_l-\max(q_l)+1, \max(q_l))$, clearly generating two terms $(i_s,j_s),(i_{s+1},j_{s+1}) \in \pathR$ such that $i_{s+1} < i_s$ while $j_{s+1}=j_s+i_s-i_{s+1}$ thus, satisfying another condition. 

Consider two neighbouring terms $(p_l,q_l), (p_{l+1},q_{l+1}) \in \coarseAlignD$. Suppose the last (or the only) term generated for $(p_l,q_l)$ by the algorithm is $(i_r,j_r)=(\max(p_l)-\max(q_l)+1, \max(q_l))$ then the first (or the only) term generated for $(p_{l+1},q_{l+1})$ will be $(i_{r+1},j_{r+1})=(\min(p_{l+1})-\min(q_{l+1})+1, \min(q_{l+1}))$. According to the definition of a \textit{coarse alignment}, we have $p_l \cap p_{l+1}=\emptyset$ and $p_l <_e p_{l+1}$ which means that $\min(p_{l+1})=\max(p_l)+1$. Also the condition $q_l \cap q_{l+1}=\emptyset$ and $q_l <_e q_{l+1}$ implies $\min(q_{l+1})=\max(q_l)+1$ thus making sure that $j_{r+1}=j_r+1$. Combining these two conditions we get that $i_{r+1}=i_r$. Suppose, $i_{r+1}=i_r$, then $\max(p_l)+1-\max(q_l)=\min(p_{l+1})-\min(q_{l+1})+1$ which would imply that $\min(q_{l+1})=1+\max(q_l)$ proving $j_{r+1}=j_r+1$. Thus, satisfying another two conditions for a \textit{path}. 

Therefore, we see that if the input is a coarse alignment $\coarseAlignD \in \mathcal{C}_{3L-2,L}$ containing only good terms, then the Algorithm \ref{Alg:CtoPmain} generates a path $\pathR \in \setOfPaths_{2L-1,L}$.
\end{proof}

We thus prove Lemma \ref{thm:AppEitherPathorD} using Fact \ref{thm:AppFactDlessC} and Fact \ref{thm:AppFactGoodtermImpliesPath}.
\end{proof}

\section{Query Lower bound for the \texorpdfstring{$\propPedit$}{Pedit} property}
\label{sec:AppQueryComplexity}

\begin{figure}[h]
\centering
\begin{tikzpicture}
    [
        box/.style={rectangle,draw=gray, minimum size=0.5cm},
    ]

\foreach \x in {-2,-1.5,...,1.5}{
    \foreach \y in {-4,-3.5,...,3}
        \node[box, fill=lightgray] at (\x,\y){0};
}
\pgfmathsetseed{1};

\foreach \x in {-2,-1.5,...,1.5}{
    \pgfmathparse{int(random(0,1))}\let\vara=\pgfmathresult;
    \node[box, fill=white] at (\x,-0.5){\vara};
}

\foreach \x in {-1.5,-1,...,1.5}{
    \pgfmathparse{int(random(0,1))}\let\vara=\pgfmathresult;
    \node[box, fill=white] at (\x,0){\vara};
}

\foreach \x in {-1,-0.5,...,1.5}{
    \pgfmathparse{int(random(0,1))}\let\vara=\pgfmathresult;
    \node[box, fill=white] at (\x,0.5){\vara};
}

\foreach \x in {-0.5,0,...,1.5}{
    \pgfmathparse{int(random(0,1))}\let\vara=\pgfmathresult;
    \node[box, fill=white] at (\x,1){\vara};
}

\foreach \x in {0,0.5,...,1.5}{
    \pgfmathparse{int(random(0,1))}\let\vara=\pgfmathresult;
    \node[box, fill=white] at (\x,1.5){\vara};
}

\foreach \x in {0.5,1,...,1.5}{
    \pgfmathparse{int(random(0,1))}\let\vara=\pgfmathresult;
    \node[box, fill=white] at (\x,2){\vara};
}

\foreach \x in {1,1.5,...,1.5}{
    \pgfmathparse{int(random(0,1))}\let\vara=\pgfmathresult;
    \node[box, fill=white] at (\x,2.5){\vara};
}

\foreach \x in {1.5,...,1.5}{
    \pgfmathparse{int(random(0,1))}\let\vara=\pgfmathresult;
    \node[box, fill=white] at (\x,3){\vara};
}

\foreach \x in {-2,-1.5,...,1}{
    \pgfmathparse{int(random(0,1))}\let\vara=\pgfmathresult;
    \node[box, fill=white] at (\x,-1){\vara};
}

\foreach \x in {-2,-1.5,...,0.5}{
    \pgfmathparse{int(random(0,1))}\let\vara=\pgfmathresult;
    \node[box, fill=white] at (\x,-1.5){\vara};
}

\foreach \x in {-2,-1.5,...,0}{
    \pgfmathparse{int(random(0,1))}\let\vara=\pgfmathresult;
    \node[box, fill=white] at (\x,-2){\vara};
}

\foreach \x in {-2,-1.5,...,-0.5}{
    \pgfmathparse{int(random(0,1))}\let\vara=\pgfmathresult;
    \node[box, fill=white] at (\x,-2.5){\vara};
}

\foreach \x in {-2,-1.5,...,-1}{
    \pgfmathparse{int(random(0,1))}\let\vara=\pgfmathresult;
    \node[box, fill=white] at (\x,-3){\vara};
}

\foreach \x in {-2,-1.5,...,-1.5}{
    \pgfmathparse{int(random(0,1))}\let\vara=\pgfmathresult;
    \node[box, fill=white] at (\x,-3.5){\vara};
}

\foreach \x in {-2,...,-2}{
    \pgfmathparse{int(random(0,1))}\let\vara=\pgfmathresult;
    \node[box, fill=white] at (\x,-4){\vara};
}

\draw[color=blue, thick] (-1.25,-1.25)--(0.75,-1.25);
\draw[color=blue, thick] (-1.25,0.75)--(0.75,0.75);
\draw[color=blue, thick] (0.75,-1.25)--(0.75,0.75);
\draw[color=blue, thick] (-1.25,-1.25)--(-1.25,0.75);
\end{tikzpicture}

\caption{\small{Example of a boolean matrix $\matrixM$ of size $(2^{n/2+1}-1) \times 2^{n/2}$ and its corresponding sub-matrix $\matrixM'$ (highlighted in blue) of size $2^{n/2-1} \times 2^{n/2-1}$. Notice that all the cells in gray are 0s.}}
\label{fig:MatrixM}
\end{figure}

\begin{figure}[h]
\centering
\begin{tikzpicture}
    [
        box/.style={rectangle,draw=black, minimum size=0.5cm},
    ]

\foreach \x in {-2,-1.5,...,1.5}{
    \foreach \y in {-4,-3.5,...,3}
        \node[box, fill=lightgray] at (\x,\y){0};
}
\pgfmathsetseed{1};

\foreach \x in {-2,-1.5,...,1.5}{
    \pgfmathparse{int(random(0,1))}\let\vara=\pgfmathresult;
    \node[box, fill=white] at (\x,-0.5){0};
}

\foreach \x in {-1.5,-1,...,1.5}{
    \pgfmathparse{int(random(0,1))}\let\vara=\pgfmathresult;
    \node[box, fill=white] at (\x,0){0};
}

\foreach \x in {-1,-0.5,...,1.5}{
    \pgfmathparse{int(random(0,1))}\let\vara=\pgfmathresult;
    \node[box, fill=white] at (\x,0.5){0};
}

\foreach \x in {-0.5,0,...,1.5}{
    \pgfmathparse{int(random(0,1))}\let\vara=\pgfmathresult;
    \node[box, fill=white] at (\x,1){0};
}

\foreach \x in {0,0.5,...,1.5}{
    \pgfmathparse{int(random(0,1))}\let\vara=\pgfmathresult;
    \node[box, fill=white] at (\x,1.5){0};
}

\foreach \x in {0.5,1,...,1.5}{
    \pgfmathparse{int(random(0,1))}\let\vara=\pgfmathresult;
    \node[box, fill=white] at (\x,2){0};
}

\foreach \x in {1,1.5,...,1.5}{
    \pgfmathparse{int(random(0,1))}\let\vara=\pgfmathresult;
    \node[box, fill=white] at (\x,2.5){0};
}

\foreach \x in {1.5,...,1.5}{
    \pgfmathparse{int(random(0,1))}\let\vara=\pgfmathresult;
    \node[box, fill=white] at (\x,3){0};
}

\foreach \x in {-2,-1.5,...,1}{
    \pgfmathparse{int(random(0,1))}\let\vara=\pgfmathresult;
    \node[box, fill=white] at (\x,-1){0};
}

\foreach \x in {-2,-1.5,...,0.5}{
    \pgfmathparse{int(random(0,1))}\let\vara=\pgfmathresult;
    \node[box, fill=white] at (\x,-1.5){0};
}

\foreach \x in {-2,-1.5,...,0}{
    \pgfmathparse{int(random(0,1))}\let\vara=\pgfmathresult;
    \node[box, fill=white] at (\x,-2){0};
}

\foreach \x in {-2,-1.5,...,-0.5}{
    \pgfmathparse{int(random(0,1))}\let\vara=\pgfmathresult;
    \node[box, fill=white] at (\x,-2.5){0};
}

\foreach \x in {-2,-1.5,...,-1}{
    \pgfmathparse{int(random(0,1))}\let\vara=\pgfmathresult;
    \node[box, fill=white] at (\x,-3){0};
}

\foreach \x in {-2,-1.5,...,-1.5}{
    \pgfmathparse{int(random(0,1))}\let\vara=\pgfmathresult;
    \node[box, fill=white] at (\x,-3.5){0};
}

\foreach \x in {-2,...,-2}{
    \pgfmathparse{int(random(0,1))}\let\vara=\pgfmathresult;
    \node[box, fill=white] at (\x,-4){0};
}

\foreach \x in {-1,-0.5,...,0.5}{
    \foreach \y in {-1,-0.5,...,0.5}{
    \pgfmathparse{int(random(0,1))}\let\vara=\pgfmathresult;
    \node[box, fill=gray] at (\x,\y){\vara};
    }
}

\end{tikzpicture}

\caption{\small{Example of a boolean matrices in $\mathcal{M}$ and its corresponding sub-matrix $\matrixM'$ (highlighted in dark gray). Notice that all the cells in the white and light gray region are 0s.}}
\label{fig:MatrixMathcalM}
\end{figure}

\begin{theorem}
\label{thm:AppPeditQuery}
The bounded-error quantum query complexity for computing the property $\propPedit$ on matrices of size $(2^{{n/2}+1}-1) \times 2^{{n/2}}$ that are $\encodingName$ of truth tables of non-deterministic branching programs with $n$ input variables is $\Omega(2^{0.75n})$. \end{theorem}
\begin{proof}
The matrices that are $\encodingName$\footnote{Refer to $\encodingName$ mentioned in Definition \ref{def:BP-P}.} of truth tables of non-deterministic branching programs with $n$ input variables (for example, matrix $\matrixM$) are of the form shown in Figure \ref{fig:MatrixM}. We prove the quantum query lower bound of $\propPedit$ on these matrices using the quantum adversary method by \cite{QLowerBounds-Ambainis-00}, but, instead of analyzing the matrix $\matrixM$ of size $2N-1 \times N$ (where $N=2^{{n/2}}$) we analyze the sub-matrix $\matrixM'$ of size $\frac{N}{2} \times \frac{N}{2}$ as shown in the Figure \ref{fig:MatrixM}. 

Let $\mathcal{M}=\{\matrixM|\matrixM \text{ are matrices of size $2N-1 \times N$ that are of the form shown in}$ $\text{ Figure \ref{fig:MatrixMathcalM}}\}$.\\ Computing the property $\propPedit$ on  matrices $\matrixM \in \mathcal{M}$ for a threshold $T_r=\frac{3N}{4}C_0+\frac{N}{4}C_1$ is equivalent to computing $\propPedit$ on the sub-matrices $\matrixM'$ (corresponding to each $\matrixM$) of size $\frac{N}{2} \times \frac{N}{2}$ for a different threshold value $V=\frac{N}{4} C_0+\frac{N}{4}C_1$ where the problem is to decide whether a minimum-cost path in $\matrixM'$ has its cost $< V$.\footnote{Using a simple geometrical argument one can prove that a square matrix $\matrixM'$ of size $\frac{N}{2}\times\frac{N}{2}$ fits inside the white region of matrix $\matrixM$ of size $2N-1 \times N$, refer to Figure \ref{fig:MatrixM} or Figure \ref{fig:MatrixMathcalM}.} 

Recall that the property $\propPedit: \{0,1\}^{\frac{N}{2}} \times \{0,1\}^{\frac{N}{2}} \times \{0,1,\dots,Q\} \rightarrow \{0,1\}$ (Definition \ref{def:PP_EditProperty}) is a function of a matrix $\matrixM$ and a parameter $\mu$ which is only required to calculate the \emph{jump} costs. However, the adversarial sets $X$ and $Y$ that we define to compute the quantum query complexity of the property $\propPedit$ doesn't need any reference to the \emph{jump} costs. Therefore, for the sake of simplicity of the proof we just define the property to be $\propPedit: \{0,1\}^{\frac{N}{2}} \times \{0,1\}^{\frac{N}{2}} \rightarrow \{0,1\}$.

\paragraph{Building the adversarial sets $X$ and $Y$:}
We choose a relation $R \subseteq X \times Y \subseteq \propPedit^{-1}(0) \times \propPedit^{-1}(1)$ where $X=\{x| (x,y) \in R\}$ and $Y=\{y| (x,y) \in R\}$. The relation $R$ is chosen such that, 
\begin{enumerate}
    \item For each matrix $x \in X$, each row in the matrix $x$ has exactly $\frac{N}{4}$ number of $1$s. Which implies that for all $x \in X$, the $\pathCost$ for each row is equal to $\frac{N}{4}C_0+\frac{N}{4}C_1=V$. We now have to ensure that for all matrices $x \in X$, the minimum-costing path also has its cost greater than or equal to $V$, which is addressed in item 3.
    \item For each matrix $y \in Y$, there is only one row in the matrix $y$ that has exactly $(\frac{N}{4}+1)$ number of $1$s and rest of the rows in $y$ have exactly $\frac{N}{4}$ number of $1$s. Which implies that for all $y \in Y$, there is a row whose $\pathCost$ is $(\frac{N}{4}-1)C_0+(\frac{N}{4}+1)C_1 < V$, as $C_0 > C_1$ (Refer to Definition \ref{def:PP_EditProperty}). 
    \item We impose additional constraints in building the sets $X$ and $Y$ so that considering only paths without any \emph{jumps} is enough to decide whether or not the min-cost path has its cost less than $V$. We introduce a set of symbols that will be useful in understanding our construction. 
    
    For a chosen even number $k$, we construct some symbols recursively in the following way: 
\begin{enumerate}
    \item $0_i \subseteq \{+_{i-1}, -_{i-1}, 0_{i-1}\}^k$ such that, number of  $+_{i-1}$-type symbols in $0_i$ is equal to number of  $-_{i-1}$-type symbols in $0_i$.
    \item $+_i \subseteq \{+_{i-1}, -_{i-1}, 0_{i-1}\}^k$ such that, number of  $+_{i-1}$-type symbols in $+_i$ is one more than number of  $-_{i-1}$-type symbols in $+_i$, but overall there is only one 1 more than the number of 0s.
    \item $-_i \subseteq \{+_{i-1}, -_{i-1}, 0_{i-1}\}^k$ such that, number of  $+_{i-1}$-type symbols in $-_i$ is one less than number of  $-_{i-1}$-type symbols in $-_i$, but overall there is only one 0 more than the number of 1s.
\end{enumerate}

Base case symbols $\{0_{0}, +_{0}, -_{0},\}$ are defined as follows:
\begin{enumerate}
    \item $0_0 \subseteq \{0,1\}^k$ such that, number of 1s in $0_0$ is equal to number of 0s in $0_0$. Therefore
    , $k$ has to be even.
    \item $+_0 \subseteq \{0,1\}^k$ such that, number of 1s in $+_0$ is two more than number of 0s in $+_0$. 
    \item $-_0 \subseteq \{0,1\}^k$ such that, number of 1s in $-_0$ is two less than number of 0s in $-_0$.
\end{enumerate}
For all $x \in X$, each row of the matrix $x$ contains $k$ symbols from the set $\{+_t,-_t,0_t\}$ such that number of $+_t$ symbols is equal to the number of $-_t$ symbols. While for all $y \in Y$, only one row in the matrix $y$ contains one more of $+_t$ symbol when compared to the number of $-_t$ symbols. The rest of the rows are balanced just like the rows of the matrices belonging to set $X$. Example of such a matrices can be viewed in Figure \ref{fig:gridAdversaryX}. 

The choice of $k$ and $t$ will be such that (1) $k^{t+2}=\frac{N}{2}$, condition that ensures that number of elements in each row of the matrices is $\frac{N}{2}$ and (2) $C_1\cdot k \cdot (t+2) < C_{jump}$, a condition that ensures that the reduction in the $\pathCost$ is less than the \emph{jump} cost, because of which we don't need to consider paths with jumps for computing the query complexity of $\propPedit$ on matrices in $X$ and $Y$. These conditions will be reviewed again in the later parts of the proof where we calculate the query lower bound of $\propPedit$ in the $\epsilon$-bounded error setting.
\end{enumerate}

\begin{figure}[h]
\centering
\begin{tikzpicture}
\draw[step=0.5cm,color=gray] (-1,-1) grid (1,1);

\node at (-0.75,+0.25) {$\textcolor{blue}{-}$};
\node at (-0.25,+0.25) {$\textcolor{gray}{0}$};
\node at (+0.25,+0.25) {$\textcolor{gray}{0}$};
\node at (+0.75,+0.25) {$\textcolor{red}{+}$};

\node at (-0.75,-0.75) {$\textcolor{gray}{0}$};
\node at (-0.25,-0.75) {$\textcolor{blue}{-}$};
\node at (+0.25,-0.75) {$\textcolor{gray}{0}$};
\node at (+0.75,-0.75) {$\textcolor{red}{+}$};

\node at (-0.75,+0.75) {$\textcolor{red}{+}$};
\node at (-0.25,+0.75) {$\textcolor{red}{+}$};
\node at (+0.25,+0.75) {$\textcolor{blue}{-}$};
\node at (+0.75,+0.75) {$\textcolor{blue}{-}$};

\node at (-0.75,-0.25) {$\textcolor{red}{+}$};
\node at (-0.25,-0.25) {$\textcolor{blue}{-}$};
\node at (+0.25,-0.25) {$\textcolor{red}{+}$};
\node at (+0.75,-0.25) {$\textcolor{blue}{-}$};


\draw[step=0.5cm,color=gray] (2,-1) grid (4,1);
\draw[-,color=gray] (2,-1) grid (2,1);


\node at (2.25,+0.75) {$\textcolor{red}{+}$};
\node at (2.75,+0.75) {$\textcolor{red}{+}$};
\node at (3.25,+0.75) {$\textcolor{blue}{-}$};
\node at (3.75,+0.75) {$\textcolor{blue}{-}$};

\node at (2.25,+0.25) {$\textcolor{blue}{-}$};
\node at (2.75,+0.25) {$\textcolor{gray}{0}$};
\node at (3.25,+0.25) {$\textcolor{gray}{0}$};
\node at (3.75,+0.25) {$\textcolor{red}{+}$};

\node at (2.25,-0.25) {$\textcolor{red}{+}$};
\node at (3.25,-0.25) {$\textcolor{red}{+}$};
\node at (3.75,-0.25) {$\textcolor{blue}{-}$};

\node at (2.75,-0.25) {$\textbf{\textcolor{cyan}{0}}$};

\node at (2.25,-0.75) {$\textcolor{gray}{0}$};
\node at (2.75,-0.75) {$\textcolor{blue}{-}$};
\node at (3.25,-0.75) {$\textcolor{gray}{0}$};
\node at (3.75,-0.75) {$\textcolor{red}{+}$};


\end{tikzpicture}
\caption{\small{Example of a boolean matrix $x \in X$ (left) and a boolean matrix $y \in Y$ (right) such that $(x,y) \in R$. The symbol $'+'=+_t$, $'0'=0_t$ and $'-'=-_t$.}}
\label{fig:gridAdversaryX}
\end{figure}

\paragraph{The quantum query complexity of $\propPedit$}
\begin{enumerate}
    \item For each $x \in X$, we have at least $(\frac{k}{2})^{t+2}\cdot \frac{N}{2}$ number of $y$s, such that $(x,y) \in R$.
    \item For each $y \in Y$, we have at least $(\frac{k}{2})^{t+2}$ number of $x$s, such that $(x,y) \in R$. 
    \item We can visualize a matrix of size $\frac{N}{2} \times \frac{N}{2}$ as a string of length $\frac{N^2}{4}$. For each $x \in X$ and $i \in [\frac{N^2}{4}]$, there is only one input $y \in Y$ such that $(x,y) \in R$ and $x_i \ne y_i$.
    \item For each $y \in Y$ and $i \in [\frac{N^2}{4}]$, there is only one input $x \in X$ such that $(x,y) \in R$ and $x_i \ne y_i$.
    \item Therefore, the quantum adversary method by \cite{QLowerBounds-Ambainis-00} gives a quantum lower bound of\\ $\Omega(\frac{k^{(t+2)}}{2^{(t+2)}}\cdot \sqrt{N})$ for distinguishing the sets $X$ and $Y$. 
    \item As mentioned earlier we chose the values $k$ and $t$ such that matrices in both these sets $X$ and $Y$ have their optimal paths without any jumps. The maximum number of ones that can be gained with a jump is $k\cdot (t+1)+\frac{k}{2} < k\cdot (t+2)$. The reduction in the traversal cost due to these $1$s will be less than $C_1\cdot k\cdot (t+2)$. Therefore, as long as we chose the values $k$ and $t$ such that $C_1\cdot k\cdot (t+2) < C_{jump}$ there will not be any jump whatsoever.
    \item We have to chose the values of $k$ and $t$ such that the lower bound mentioned in item 5 is maximised while satisfying the two following constraints: (1) $k^{t+2}=\frac{N}{2}$ and (2) $C_1\cdot k\cdot (t+2) < C_{jump}$. From (1) we get $t+2=\log_{k}(N/2)$, hence, $k\cdot(t+2)=\frac{k}{\log k}\log(\frac{N}{2})$. Therefore, the lower bound is $\Omega(\frac{N}{2}^{(1.5-\frac{1}{\log k})})$. By fixing $k=\omega(1)$, we get the lower bound of $\Omega(N^{1.5})$.\footnote{Recall from Definition \ref{def:PP_EditProperty} that $C_0, C_1$ are constants and $C_{jump}=\omega(\log N)$.} Therefore a lower bound of $\Omega(2^{0.75n})$ as $N=2^{n/2}$.
\end{enumerate}
Therefore, we can conclude that the bounded-error quantum query lower bound for computing $\propPedit$ on matrices that are $\encodingName$ of truth tables of non-deterministic branching programs with $n$ input variables is $\Omega(2^{0.75n})$.
\end{proof}

\begin{corollary}
\label{thm:AppCorQofPPedit}
The bounded-error quantum query complexity for computing $\propertyPPedit$ on matrices of size $(2^{n/2+1}-1) \times 2^{n/2}$ that are $\encodingName$ of truth tables of non-deterministic branching programs with $n$ input variables from the set $\mathcal{S}$ is $\Omega(2^{0.75n})$. Here the set $\mathcal{S}$ is
\begin{equation*}
\mathcal{S}=\{S \text{ }|\text{ } \matrixM^{\truthtable(S)} \in \propertyPPedit^{-1}(0) \cup \propertyPPedit^{-1}(1)\},
\end{equation*}
where $S$ denotes the branching program and $\matrixM^{\truthtable(S)}$ denotes the $\encodingName$ of the truth table of $S$.
\end{corollary}
\begin{proof}
The property $\propertyPPedit$ defined at \ref{def:PP_EditProperty} is a promise version of property $\propPedit$. The results of Theorem \ref{thm:AppPeditQuery} also hold for the property $\propertyPPedit$ because the adversarial sets that we construct in that theorem doesn't depend on the value of the parameter $\mu$. Therefore, for a constant $\epsilon$, $\Query(\propertyPPedit|_{\mathcal{S}})=\Omega(2^{0.75n})$. 
\end{proof}

\end{appendices}

\end{document}